\documentclass[12pt]{article}

\usepackage{amsthm}
\usepackage{amsmath}
\usepackage{thmtools,thm-restate}
\usepackage{amssymb}
\usepackage{graphicx}
\usepackage[margin=2cm, includefoot]{geometry}
\usepackage{commath}
\usepackage{xcolor}
\usepackage{color}
\usepackage{url}
\usepackage{algorithm}
\usepackage{algpseudocode}
\usepackage{bm}
\usepackage{verbatim}
\usepackage{subcaption}
\usepackage{enumitem}

\newcommand{\I}{\iota}
\newcommand{\R}{\mathbb{R}}

\newcommand{\hx}{\hat{x}}
\newcommand{\E}{\mathbb{E}}

\newcommand{\X}{\mathcal{X}}

\newcommand{\SNR}{\text{SNR}}

\newcommand{\rev}[1]{{\color{black}{#1}}}

\newtheorem{thm}{Theorem}
\numberwithin{equation}{section}
\numberwithin{thm}{section} 

\newtheorem{conj}[thm]{Conjecture}
\newtheorem{cor}[thm]{Corollary}
\newtheorem{prop}[thm]{Proposition}
\newtheorem{lemma}[thm]{Lemma}

\usepackage{authblk}

\begin{document}
\title{Super-resolution multi-reference alignment}
\author[1]{Tamir Bendory}
\author[2]{Ariel Jaffe}
\author[3]{William Leeb}
\author[4]{Nir Sharon}
\author[5]{Amit Singer}
\affil[1]{School of Electrical Engineering, Tel Aviv University}
\affil[2]{Applied Mathematics Program, Yale University}
\affil[3]{School of Mathematics, University of Minnesota, Twin Cities}
\affil[4]{School of Mathematical Sciences, Tel Aviv University}
\affil[5]{Department of Mathematics and Program in Applied and Computational Mathematics, Princeton University}

\maketitle

\begin{abstract}
We study super-resolution multi-reference alignment, the problem of estimating a signal from many circularly shifted, down-sampled, and noisy observations. We focus on the low SNR regime, and show that a signal in $\R^M$ is uniquely determined when the number $L$ of samples {per observation} is of the order of the square root of the signal's length ($L=O(\sqrt{M})$). Phrased more informally, one can square the resolution. This result holds if the number of observations is proportional to  $1/\SNR^3$. In contrast, with fewer observations recovery is impossible even when the observations are not down-sampled ($L=M$). The analysis combines tools from statistical signal processing and invariant theory. We design an expectation-maximization algorithm and demonstrate that it can super-resolve the signal in challenging SNR regimes.
\end{abstract}


\section{Introduction}
\paragraph{Model.} We study the problem of estimating a signal from its circularly shifted, sampled, and noisy copies. 
More precisely, we  consider $N$ independent observations sampled from the model
\begin{equation} \label{eq:model}
y = PR_s x + \varepsilon, \qquad  s\sim \text{Uniform}[0,\ldots,M-1],\qquad \varepsilon\sim \mathcal{N}(0,\sigma^2 I),
\end{equation}
where $R_s$ denotes an operator that  shifts the target signal $x\in\R^M$ circularly by $s$ entries, that is, $(R_s x)[n]= x[(n-s)\bmod M]$, and   $P$ denotes a fixed sampling operator that collects $L\leq M$ equally-spaced samples. 
We assume that the random variable $s$ is  distributed uniformly {over $[0, \ldots, M-1]$}, and the noise $\varepsilon\in \R^L$  is   i.i.d.\ Gaussian.
Explicitly, the $i$-th observation reads: 
\begin{eqnarray} 
y_i[\ell] &=& P(R_{s_i}x)[\ell] + \varepsilon_i[\ell]
\nonumber \\ &=& x[\ell K-s_i]+\varepsilon_i[\ell],
\end{eqnarray}
where $\ell=0,\ldots,L-1$, and 
$K:=M/L$ is assumed to be an integer. 
{Importantly, the shifts $s_i$ are all unknown, and thus~\eqref{eq:model} is a special case of the multi-reference alignment (MRA) model, which we review in Section~\ref{sec:background}.}
{Figure~\ref{fig:measurement_example} presents an example of two observations with signal-to-noise ratio (SNR) \rev{equal to  one} (namely, the expected squared norm of the noise equals the squared norm of the signal).}

Our goal is to estimate $x$ from $N$ observations sampled from~\eqref{eq:model}.
{In contrast to previous works on MRA, the individual observations are down-sampled, and therefore recovering the full signal $x$ is also a special case of the super-resolution problem. Accordingly,} we refer to $x$ as the ``high-resolution signal,'' while $y_1,\ldots,y_N$ are the ``low-resolution observations.'' The parameter~$K$ can be thought of as a ``super-resolution factor.''
The difficulty in estimating $x$  resides chiefly in three factors: the additive noise, the unknown circular shifts (the nuisance variables of the problem), and the sampling operator. 

The statistical model~\eqref{eq:model} suffers from an intrinsic symmetry: it is invariant  under a {global} circular shift since 
{$p(y|x) = p(y|R_ix)$ for all $i=1,\ldots,M-1$.}
 In this case, we say that the goal is to recover the signal up to a {global} circular shift. More formally, the goal is to recover \emph{the orbit of~$x$}: 
\begin{equation} \label{eq:orbit_x}
G x:=\{gx \,|\,g\in G\},
\end{equation}
where $G:=\{R_0,R_1,\ldots,R_{M-1} \}$ is the group of cyclic shifts $\mathbb{Z}_M$.
However, as will be shown \rev{in Section~\ref{sec:analysis}}, without prior information on the signal, even the  orbit $Gx$ is not  identifiable from the observations, and thus prior information on $x$ is \emph{necessary} for its identification. 

\paragraph{Connection with sampling theory.} We think of the discrete signal $x\in \R^M$ as   Nyquist-rate samples of a continuous bandlimited 
signal.
Specifically, let us define a {real} signal with bandlimit~$B$ as
\begin{equation} \label{eq:continuous_signal}
x_c(t) = \sum_{k=-B}^B \hx[k]e^{2\pi \I kt}, \qquad t\in[0,1),
\end{equation}
where {$\I = \sqrt{-1}$, \rev{and $\hx$ denotes the Fourier series coefficients of $x$.}  Since $x_c$ is real, it follows that $\hx[k]=\overline{\hx[-k]}$.}
According to the well-known Shannon-Nyquist {sampling} theorem, the samples
\begin{equation}  \label{eq:continuous_signal_sampled}
\tilde x_c[m] := x_c(m/M)  = \sum_{k=-B}^B \hx[k]e^{2\pi \I km/M}, \qquad m=0,\ldots,M-1,
\end{equation}
characterize  $x_c$ uniquely when $M\geq 2B+1$. 
Model~\eqref{eq:model} is identical to rotating the discrete signal~$\tilde x_c$ on the $M$-point grid $\{m/M\}_{m=0}^{M-1}$, sampling it $L$ times, and adding noise.

With the above interpretation in mind, we identify the length of the signal in~\eqref{eq:model} with twice the signal's bandwidth, namely, $M=2B+1\approx 2B$. Thus, if $L<M$, we say that each observation is sampled \emph{below the signal's Nyquist rate}, and thus the recovery process should compensate for an aliasing distortion. To avoid aliasing, the standard signal processing approach {in many applications} is to remove the {high frequency components} before sampling~\cite{schafer1989discrete,eldar2015sampling}|namely, low-passfiltering the signal| and   then estimate a down-sampled, smooth approximation of $x$.
While this strategy is generally optimal for a single observation, this is not necessarily true when multiple observations are available. In this work,  we show that if sufficiently many observations are acquired (as a function of the noise level), then in principle it suffices to acquire only $L= O(\sqrt{M})$ samples at each observation to recover  high-resolution details,  even if the circular shifts are unknown and the noise level might be  high. 

The analogy between~\eqref{eq:model} and rotating a continuous bandlimited signal~\eqref{eq:continuous_signal} holds only when the rotations are restricted to the grid $\{m/M\}_{m=0}^{M-1}$. In Section~\ref{sec:future_work} we discuss  potential extensions to more intricate models that permit rotations over a continuous interval.

\begin{figure}
	\begin{subfigure}[ht]{0.48\columnwidth}
		\centering
		\includegraphics[width=\columnwidth]{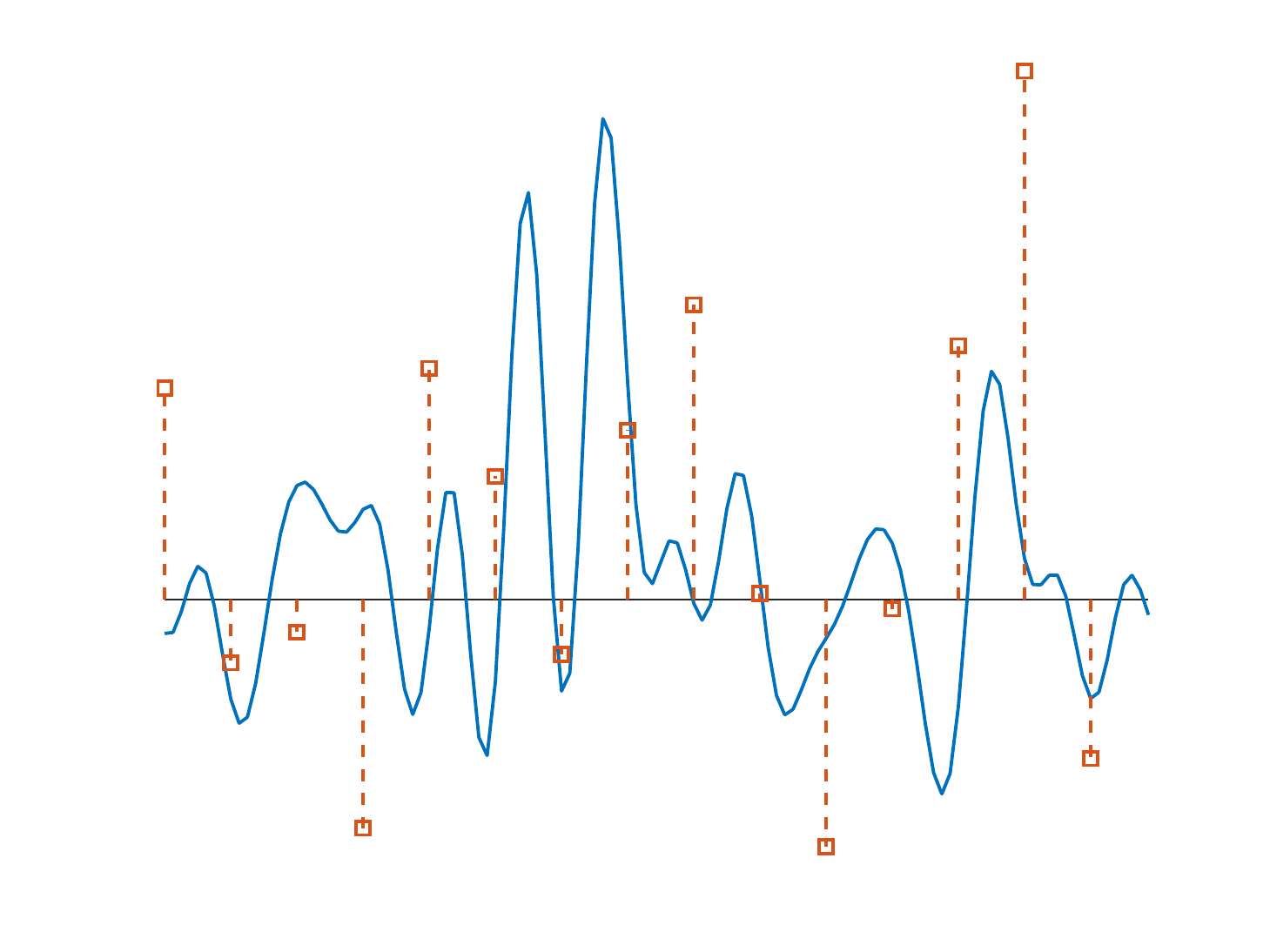}
		\caption{}
	\end{subfigure}
	\hfill
	\begin{subfigure}[ht]{0.48\columnwidth}
		\centering
		\includegraphics[width=\columnwidth]{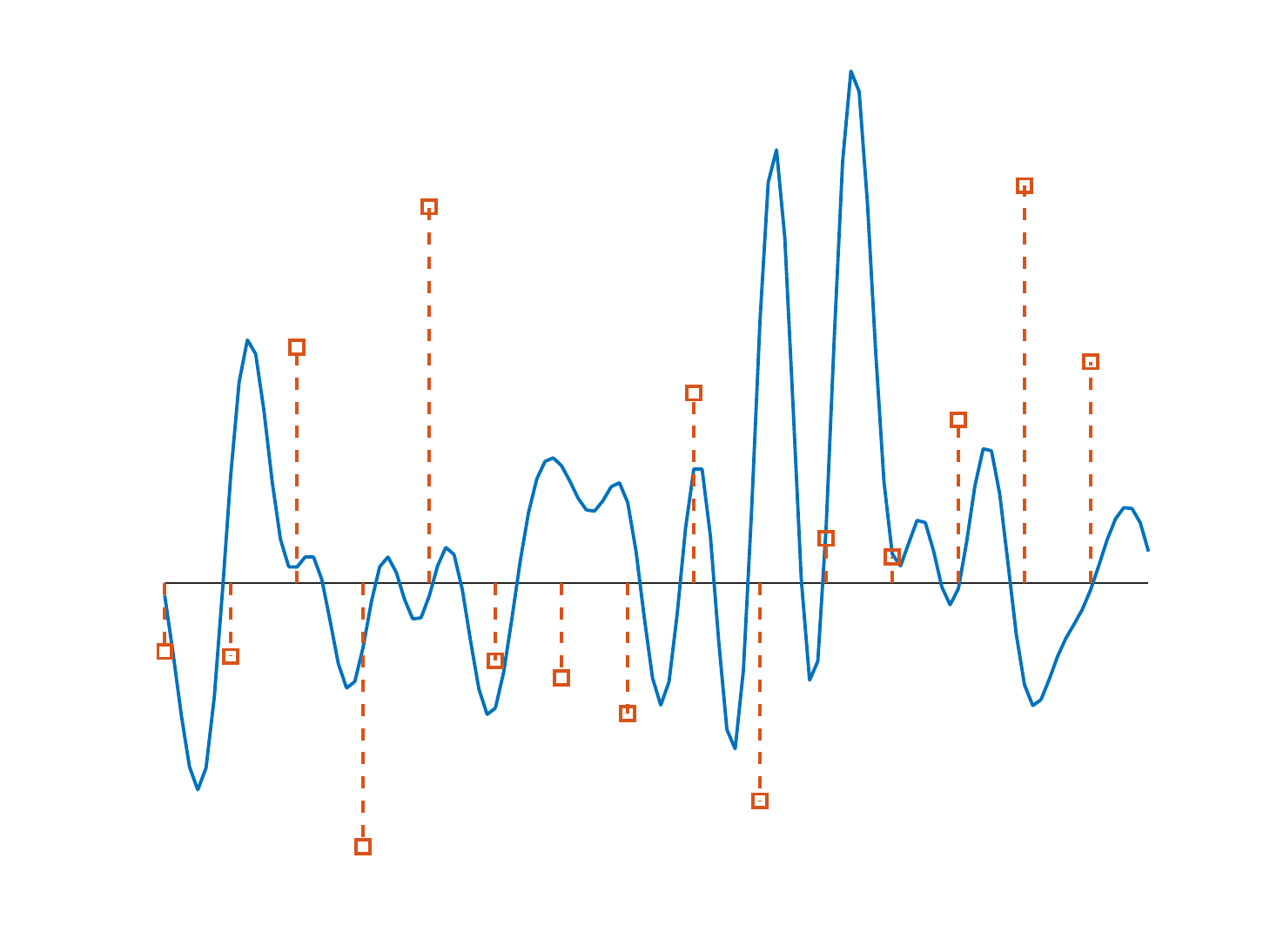}
		\caption{\label{fig:example_noisy}}
	\end{subfigure}
	\caption{	\label{fig:measurement_example}Two shifted copies of a signal  of length  $M=120$ (the high-resolution signal) are presented in blue. 
		 The red squares  display $L=15$ noisy samples   with SNR \rev{equal to one}. The goal is to estimate the high-resolution signal from {multiple noisy observations}.}
\end{figure}

\paragraph{Super-resolution.} The model~\eqref{eq:model}   is an instance of the \emph{super-resolution from multiple observations} problem: the task of estimating the fine details of a signal from its low-resolution observations. This problem has attracted the attention of numerous researchers  in the last couple of decades in a variety of fields, such as computer vision, image processing, and medical imaging; see for instance~\cite{park2003super,farsiu2004advances, greenspan2008super} and references therein.
\rev{In particular, the statistical model in some of these works is akin to~\eqref{eq:model}, see for example~\cite{robinson2006statistical,robinson2009optimal,woods2005stochastic}.}
Nevertheless, as far as we know, previous works on super-resolution did not aim to  derive and quantify the achievable super-resolution in the low SNR regime.
To avoid confusion, we mention that there exists a different thread of research, which is not directly related to this work,  that studies {super-resolution from a single image} based on prior knowledge (such as  sparsity~\cite{candes2014towards,bendory2016robust}),
or machine learning techniques~\cite{kim2016accurate,lim2017enhanced}. 

\begin{figure}[ht]
	\centering
	\includegraphics[width=.8\linewidth]{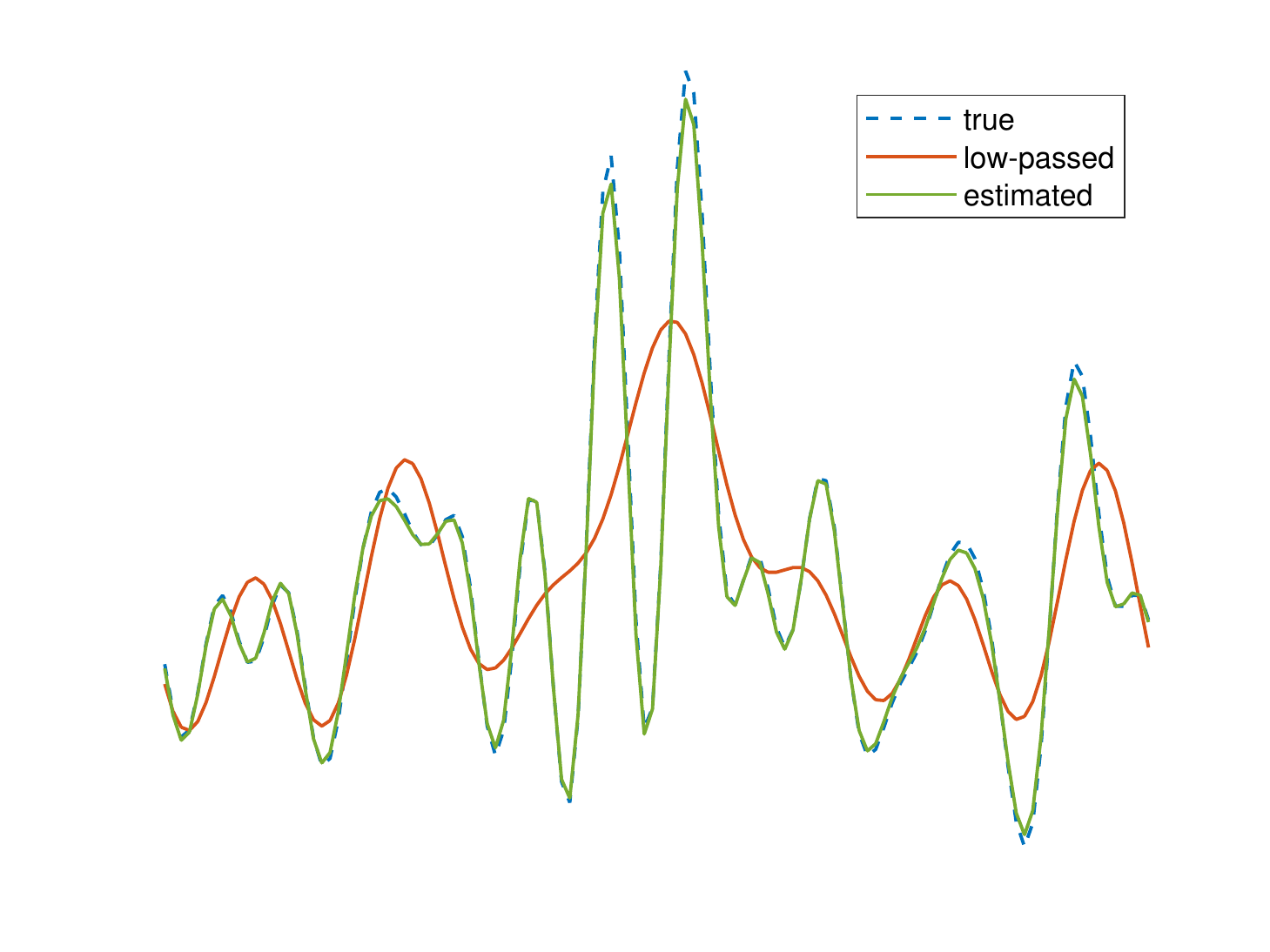}
	\caption{\label{fig:recovery_example} An example of an accurate estimate when the SNR is equal to 1. In the experiment, $N = 10^4$ observations were generated from a signal of length $M=120$ (plotted in dashed blue; the same one  as in  Figure~\ref{fig:measurement_example}). 
	The bandwidth of the signal is $B=15$ and it was sampled $L=15$ times at each observation---half of the Nyquist sampling rate.  
	The classical signal processing approach suggests to remove all frequencies beyond $L/2$ and then process the low-resolution data. This low-passed version of the signal is presented in red. Notably, the two peaks in the center of the signal are blurred and merged into one. In contrast, the EM algorithm resolves the two peaks and estimates the high-resolution signal accurately (in green).
		}
\end{figure}

\paragraph{Main contributions.}
In this paper we provide a detailed analysis of  model~\eqref{eq:model} and derive  fundamental conditions permitting an accurate estimate of $x$. In particular, we characterize the interplay between the number of observations $N$, the noise level $\sigma$, the signal's length $M$, and  the number of samples per observation $L$, in the low SNR regime.

The following \rev{theorem summarizes (informally)} the theoretical contribution of this paper. 
Precise formulations and technical details are provided in Section~\ref{sec:analysis}.

\rev{
%
\begin{thm}[informal] \leavevmode
		Suppose that $N$ observations from~\eqref{eq:model} are collected and $\sigma\to\infty$. If $N/\sigma^6\to\infty$ and $L\geq  C\sqrt{M}$ \rev{for some constant $C$},  then the maximum of the likelihood function~$p(y_1,\ldots,y_N|x)$ is attained by  a finite set of signals that {includes} the target signal~$x$. 
		If in addition  $x$ was drawn from 
a Gaussian prior, then almost surely there exists a single signal that achieves the maximum of the posterior distribution~$p(x|y_1,\ldots,y_N)$. 
\end{thm}

}

\paragraph{Expectation-maximization.} 
As a computational scheme, we propose to retrieve the high-resolution signal $x$ from the low-resolution observations $y_1,\ldots,y_N$ {using} an expectation-maximization (EM) algorithm; a detailed description {is given} in Section~\ref{sec:EM}. 
Figure~\ref{fig:recovery_example} shows a  numerical example.  A high-resolution signal of length $M=120$ is estimated  from $N=10^4$ observations in a noisy environment, {where the SNR is equal to one and each observation} is sampled at $L=15$ points. 
The bandwidth of the signal is $B=L$, so that the sampling rate is half of the Nyquist rate.  
If we were to follow the Shannon-Nyquist sampling scheme of filtering out the $L/2$ high frequencies, the two peaks in the center of the signal would have been blurred into one, even with known circular shifts and in the absence of noise.
In contrast, the EM algorithm resolves the  two adjacent peaks and estimates the signal accurately. A detailed description of this simulation, and additional numerical experiments, are provided in Section~\ref{sec:numerics}.
\rev{We note, however, that while the theoretical analysis  guarantees identifiability in the regime $M=O(L^2)$, in our experiments the EM algorithm fails to  estimate the high-resolution signal even when $L\approx M^{2/3}$.}
   Following~\cite{bandeira2017estimation,boumal2018heterogeneous,weinthesis}, we postulate that this inadequate  performance reflects a fundamental statistical-computational gap in the super-resolution problem, rather than a shortcoming of the EM framework.

\paragraph{Remark on terminology and notation.}
We refer to each {realization} of the model~\eqref{eq:model} as an observation, and to the entries of each observation as samples. 
Namely, $y_i[\ell]$ denotes the $\ell$-th sample of the $i$-th observation.
In addition, in the sequel all indices should be considered as modulo $M$ or $L$, depending on the context. 
When writing $P \gtrsim Q^d$ \rev{($\lesssim)$} we mean that $P$ is greater \rev{(smaller)} than~$Q^d$ plus a polynomial of degree $d-1$.
\rev{For example, $L\lesssim\sqrt{6M}$ implies that $M\gtrsim L^2/6 + C_1L + C_0$ for some constants $C_0$ and $C_1$.} 

\section{Background on multi-reference alignment and invariants} \label{sec:background}

The model~\eqref{eq:model} is a special case of the \emph{multi-reference alignment} (MRA) problem.
This problem entails estimating a signal  from multiple noisy observations; in each observation the signal is acted upon by an unknown element of a known group $G$. 
In its most general form, the MRA model reads
\begin{equation} \label{eq:mra}
y = T(g\circ x) +\varepsilon, \qquad g\in G, \, x\in\X,
\end{equation}
where $T$ is a known linear operator,
{with the group $G$ acting on a  vector space $\X$~\cite{bandeira2015non}.}
{Specifically,} if $x\in\R ^M$, $G$ is identified with the group of circular shifts $\mathbb{Z}_M$, and $T$ is the  sampling operator $P$, then the general MRA model~\eqref{eq:mra} reduces to~\eqref{eq:model}. 

Similarly to many MRA models in the literature~\cite{bandeira2014multireference,bendory2017bispectrum,perry2017sample,abbe2018multireference,boumal2018heterogeneous,ma2018heterogeneous,abbe2018estimation,aizenbud2019rank,romanov2020multi},  this work is  inspired by single-particle reconstruction problems using cryo-electron microscopy (cryo-EM) and X-ray free electron lasers (XFEL)---high-resolution structural methods for biological macromolecules~\cite{frank2006three,gaffney2007imaging,nogales2015cryo,singer2018mathematics,bendory2019single}. 
In particular, this work is a first step towards understanding the resolution limits of these modalities; see further discussion in Section~\ref{sec:future_work}.

Suppose we collect $N$ observations from~\eqref{eq:mra}. 
If the noise level is low, the standard approach is to estimate the group element $g_1,\ldots,g_N$.
For example, in~\eqref{eq:model}  the unknown circular shifts $s_1,\ldots,s_N$ can be estimated by simultaneous clustering  and synchronization (see Section~\ref{sec:hMRA}). This can be done, for instance, using the Non-Unique Games framework~\cite{lederman2019representation}. 
However, in the low SNR regime---which is the main interest of this work---the group elements cannot be recovered reliably by any method~\cite{aguerrebere2016fundamental,bendory2019multi}.
Therefore,  we consider two techniques that circumvent shift  determination:  estimation based on shift-invariant features, and the EM algorithm.
In particular, we formulate  EM in detail in Section~\ref{sec:EM}, and present numerical experiments  in Section~\ref{sec:numerics}. 

For the theoretical analysis,  we use features that are invariant {under} circular shifts.
Specifically, the $q$-th order circular-shift invariant feature of a signal $z\in\R^L $ is simply its auto-correlation:
\begin{equation} \label{eq:autocorrelations}
M_q(z)[\ell_1,\ldots,\ell_{q-1}]=\sum_{i=0}^{L-1} z[i]z[i+\ell_1]\ldots z[i+\ell_{q-1}].
\end{equation} 
It is readily seen that this quantity remains unchanged under any circular shift of $z$, namely, $M_q(z)=M_q(R_{\bar{s}} z)$ for any fixed $\bar{s}$. 
These invariants can also be presented in  Fourier domain.
Specifically, let ${\hat{z}[k]}$ denote the $k$-th Fourier coefficient of $z$. Then the 
\rev{polynomials}
\begin{equation}
\hat{M}_q(z)[k_1,\ldots,k_{q-1}]=\hat{z}[k_1]\cdots \hat{z}[k_{q-1}]{\hat{z}[-k_1-\cdots-k_{q-1}]},
\end{equation} 
are also invariant under circular shift. 
\rev{Throughout the paper, we use the terms auto-correlations, invariants, and invariant features interchangeably.}
Using these invariants, a variety of algorithms were proposed under different MRA setups~\cite{bendory2017bispectrum,perry2017sample,abbe2018multireference,boumal2018heterogeneous,chen2018spectral,ma2018heterogeneous}, as well as for  cryo-EM and XFEL~\cite{kam1980reconstruction,kurta2017correlations,levin20173d,bendory2018toward,pande2018ab,von2018structure,sharon2019method}.

In this work, we harness the first three invariants. The first invariant is the zero frequency $\hat{M}_1(z) = \hat{z}[0]$ (equivalently, the mean of the signal). The second invariant is the power spectrum of the signal $\hat{M}_2(z)[k]=|z[k]|^2$ for $k=0,\ldots,L-1$. Unfortunately, the mean and the power spectrum do not determine a general signal uniquely (see for example~\cite{bendory2017fourier}). 
Thus, we need the third-order invariant, the \emph{bispectrum}, which determines almost all signals uniquely~\cite{tukey1953spectral,sadler1992shift}:
\begin{equation} \label{sec:bispectrum}
\hat{M}_3(z)[k_1,k_2] = \rev{\hat{z}[k_1]\hat{z}[k_2]\hat{z}[-k_1-k_2]}, \qquad k_1,k_2=0,\ldots,L-1.
\end{equation}
The bispectrum is a useful tool in many data processing applications, such as separating Gaussian and non-Gaussian processes~\cite{brockett1988bispectral}, studying the cosmic background
radiation, seismic, radar and EEG signals~\cite{wang2000cosmic,chen2008feature,ning1989bispectral}, MIMO systems~\cite{chen2001frequency}, and classification~\cite{zhao2014rotationally}.  

For large $\sigma$, the variance of estimating  the $q$-th order auto-correlation \rev{(either ${M}_q$ or $\hat{M}_q$)} is proportional to~$\sigma^{2q}$ since the estimator involves the product of $q$ noisy terms. Thus, reliable estimation requires  an order of $\sigma^{2q}$ observations. For the problem under consideration, it implies that we need to record $N/\sigma^6\gg 1$ observations to obtain an accurate estimate of the bispectrum.
Interestingly, it was shown that for the MRA model~\eqref{eq:mra}, the invariant features approach 
is optimal in the following sense. Let $\bar{q}$ be the lowest-order auto-correlation that identifies a generic signal (in our case, $\bar{q}=3$). Then, in the asymptotic regime where $N$ and $\sigma$ diverge (while $L$ is fixed), the estimation error of any method is bounded away from zero if  $N/\sigma^{2\bar{q}}$ is bounded from above~\cite{bandeira2017optimal,abbe2018estimation}. 
In other words, $\bar{q}$ determines the minimal number of observations required for an accurate estimate in the low SNR regime.
Remarkably, we show that for~\eqref{eq:model} \rev{and for a certain range of $L$}, at the same estimation rate (i.e., $N$ scales with $\sigma^6$) one can reduce the sampling rate significantly below the Nyquist rate and still achieve an accurate estimate of the signal.
In Section~\ref{sec:future_work} we discuss the potential of super-resolution in case higher-order auto-correlations could be computed---that is, if more observations are available.

\section{Analysis} \label{sec:analysis}
\rev{The analysis is carried out in the asymptotic regime of $N\to\infty$, while the dimension $M$ remains fixed. Therefore, we assume, without of loss of generality, that $\|x\|^2=M$ so that SNR=$1/\sigma^2$.  By the term ``accurate recovery'' we mean that the recovery error drops to zero almost surely. For example, the condition $N/\sigma^6\to\infty$ ensures that we can almost surely estimate the bispectrum accurately.} 

In Section~\ref{sec:hMRA}, we show that~\eqref{eq:model} can be interpreted as the \emph{heterogeneous multi-reference alignment} (hMRA) model applied to $K$ subsets of $x$, and formulate the likelihood function of~\eqref{eq:model}. This, in turn, immediately implies that the signal is not determined uniquely from the likelihood function \rev{(a result implicit in earlier works, such as \cite{woods2005stochastic}):}
	
	\rev{	\begin{thm} 
\label{thm:nonunique}
			The likelihood function $p(y_1,\ldots,y_N|x)$ does not determine $x$  uniquely, neither its orbit under $\mathbb{Z}_M$~\eqref{eq:orbit_x}.
		\end{thm}
	}
 
	Nevertheless, \rev{the likelihood function} allows us to identify  a  family of signals which can be described as the orbit of~$x$ under a  parameterized sub-group of the permutation group; we denote this orbit by $G_{\Pi,L} x$ from reasons that will be explained later.  	
	Our analysis consists of two stages: identifying the orbit $G_{\Pi,L} x$ from the observations, and finding a unique signal in $G_{\Pi,L} x$  that maximizes the posterior distribution.  In particular, in Section~\ref{sec:recover_orbit} we use auto-correlation analysis to show that \rev{for any $L\leq 192$ satisfying $L\gtrsim \sqrt{6M}$ (more accurately, any pair $(L,M)$ satisfying~\eqref{eq:Pl})},  the orbit $G_{\Pi,L} x$ can be computed from the first three  auto-correlations of $y$\rev{; we conjecture it remains true for any $L>192$}. These  auto-correlations  can be estimated from the data  if $N/\sigma^6\to\infty$ in any SNR regime. Finally,  in Section~\ref{sec:merge_signals} we show that if~$x$ was drawn from almost any Gaussian prior on the signal, then there is a unique signal in~$G_{\Pi,L} x$ that maximizes the posterior distribution.	 
	
	\rev{The following summarizes the main results of this section:}
\rev{	
	\begin{thm} 		
	Suppose that $N\to\infty$ observations from~\eqref{eq:model} are collected,
	$N/\sigma^6\to\infty$, and that~$x$ was drawn from almost any Gaussian prior. 
	Then, for $L\leq 192$ and any $K$ that satisfies~\eqref{eq:Pl}, there exists a single signal that achieves the maximum of the posterior distribution $p(x|y_1,\ldots,y_N)$.
	\end{thm}

	\begin{conj} 		
	Suppose that $N\to\infty$ observations from~\eqref{eq:model} are collected,
	$N/\sigma^6\to\infty$, and that~$x$ was drawn from almost any Gaussian prior. 
	Then, for any fixed $M$, there exists a single signal that achieves the maximum of the posterior distribution $p(x|y_1,\ldots,y_N)$ as long as $M\leq L\cdot\mathcal{P}(L)$, where $\mathcal{P}(L)$ is given in~\eqref{eq:Pl}.
	\end{conj}

}

\subsection{Reduction to heterogeneous MRA and the likelihood function} \label{sec:hMRA}
Consider two realizations  $y_i,y_j$ generated, respectively, after shifting $x$ by $s_i$ and $s_j$, and recall that $K=M/L$ is an integer.  If $s_i - s_j = cK$ for some integer $c$, then $y_i$ is equal to a circular shift of $y_j$, with a different noise realization. 
It follows that any observation $y_i$ is a noisy and circularly-shifted realization of one of the following $K$ signals, 
\begin{equation}\label{eq:sub_signals}
x_k := \left[x[k], \,x[k+K], \,x[k+2K], \ldots,\, x[k+(L-1)K]\right], \qquad k=0,\ldots,K-1.
\end{equation}
Namely,  {$x_k[\ell] = x[k+K\ell]$} for $\ell=0,\ldots,L-1$.
We refer to $x_0,\ldots,x_{K-1}\in\R^L$ as sub-signals.
 Using this notation,  the  model~\eqref{eq:model} can be written as
\begin{equation} \label{eq:heter_mra}
y =  R_{\ell} x_{k} + \varepsilon,
\end{equation}
where $k$ is drawn uniformly at random from $\{0,\ldots,K-1\}$,   $R_{\ell}$ is a circular shift on an $L$-point grid $[0,1,\ldots,L-1]$, and $\ell$ is distributed uniformly.  
 The model~\eqref{eq:model} is thus equivalent to the 
{hMRA} model, recently studied in~\cite{perry2017sample,bandeira2017estimation,boumal2018heterogeneous,ma2018heterogeneous}, applied to the sub-signals $x_0,\ldots,x_{K-1}$.

\begin{figure}
	\begin{subfigure}[ht]{0.48\columnwidth}
		\centering
		\includegraphics[width=.9\columnwidth]{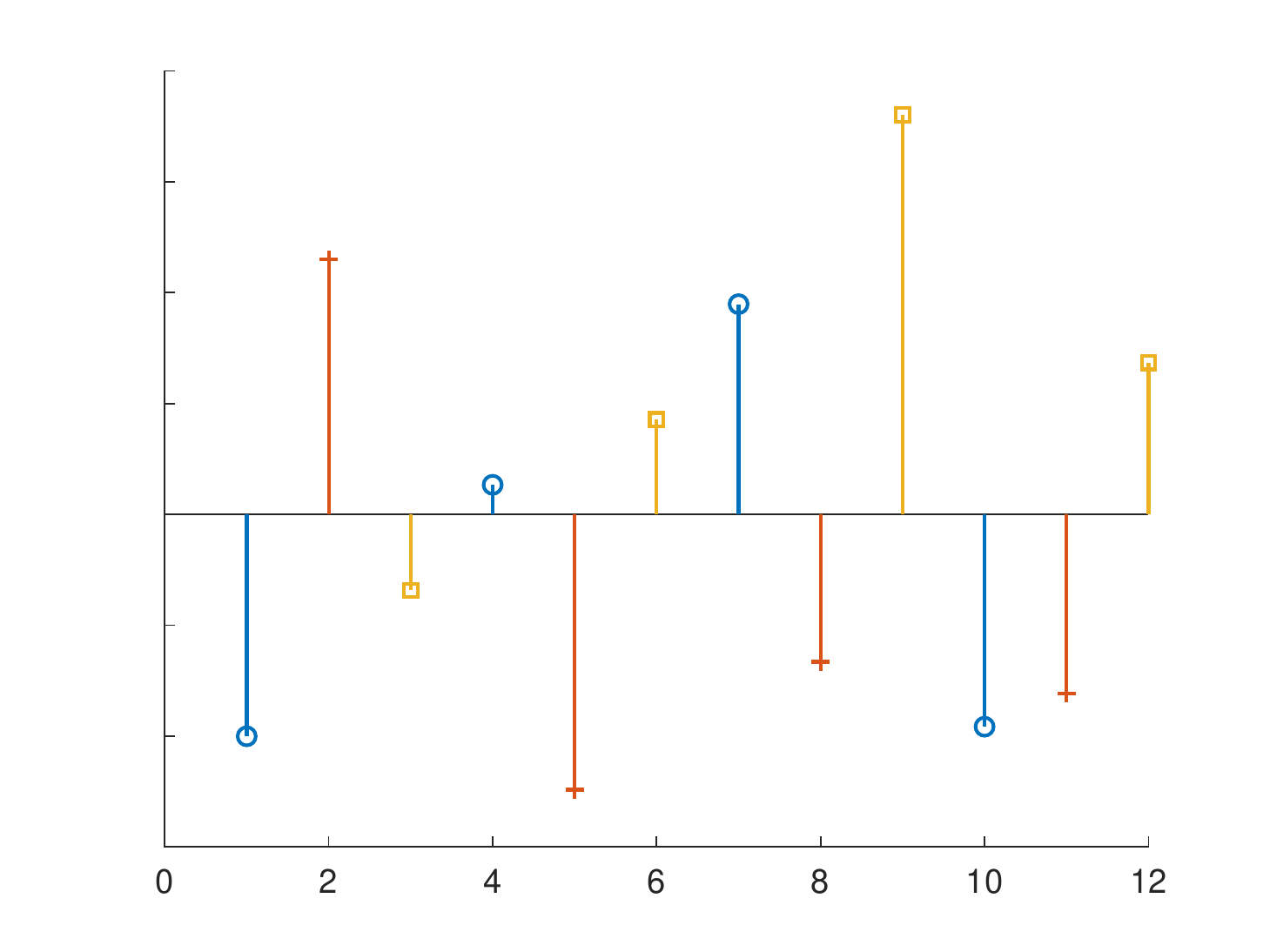}
		\caption{A signal consisting of 3 sub-signals}
	\end{subfigure}
	\hfill
	\begin{subfigure}[ht]{0.48\columnwidth}
		\centering
		\includegraphics[width=.9\columnwidth]{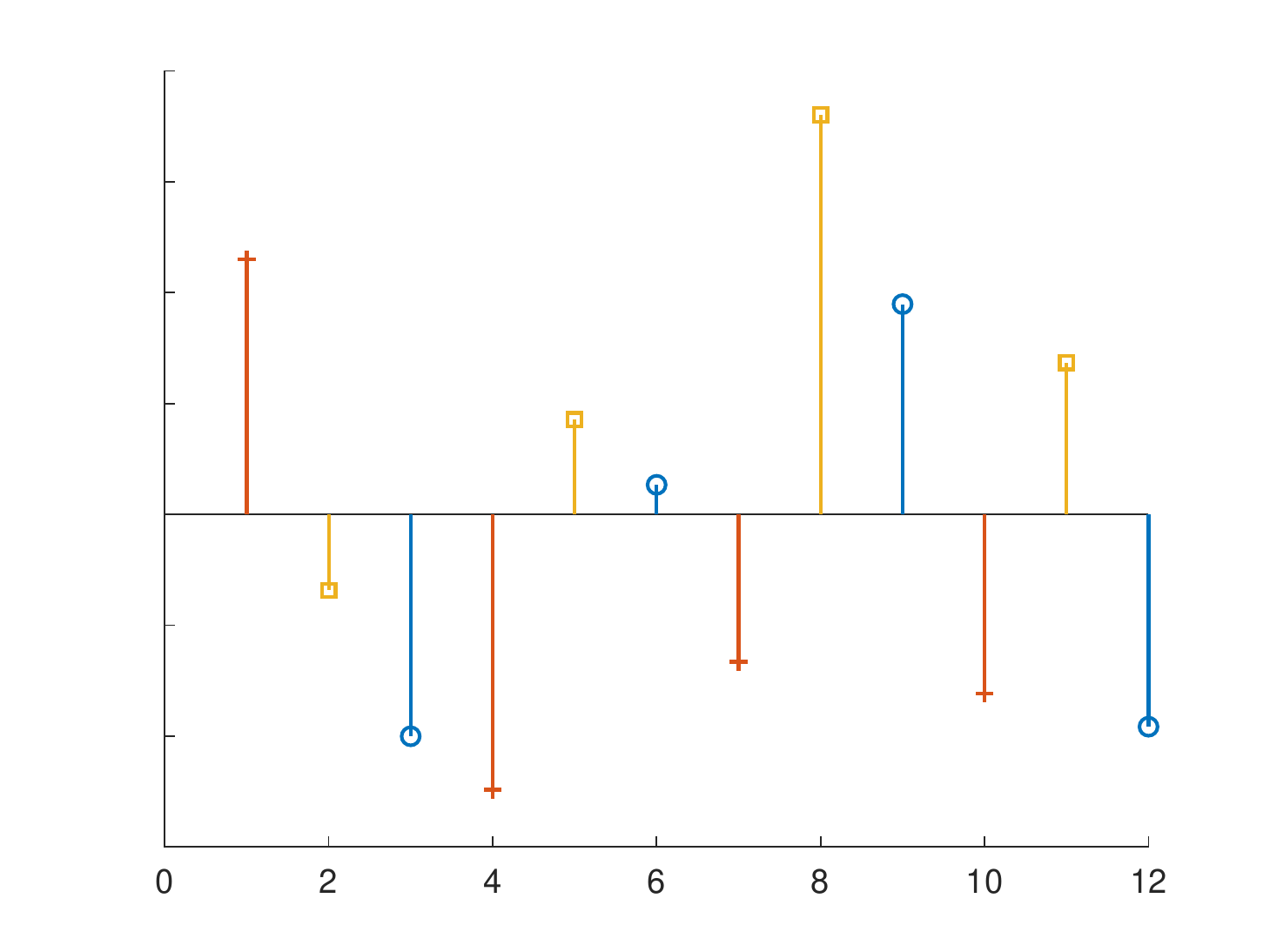}
		\caption{ A permutation of the sub-signals}
	\end{subfigure}

	\begin{subfigure}[ht]{0.48\columnwidth}
	\centering
	\includegraphics[width=.9\columnwidth]{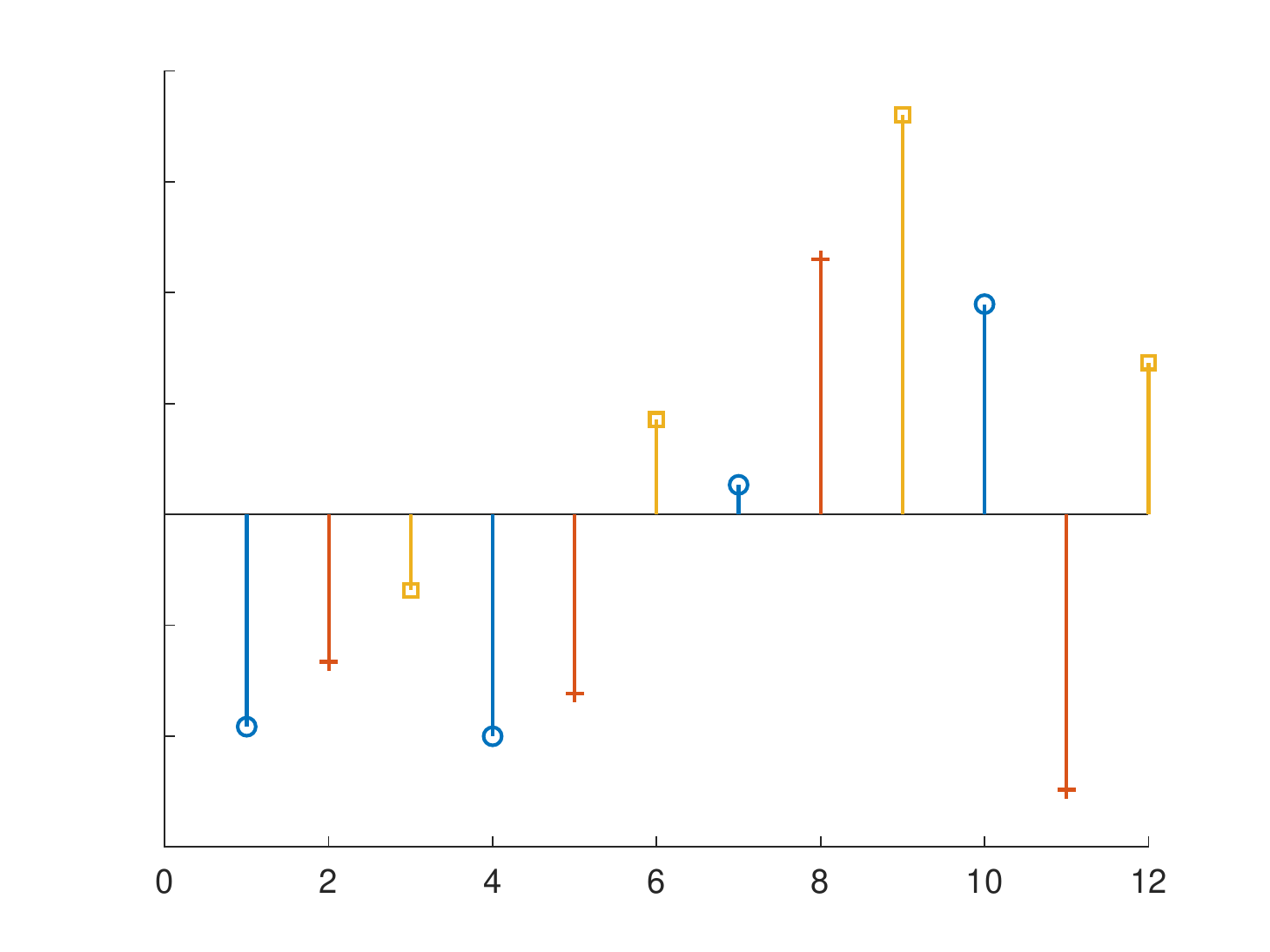}
	\caption{ Circular shifts of the sub-signals}
\end{subfigure}
\hfill
\begin{subfigure}[ht]{0.48\columnwidth}
	\centering
	\includegraphics[width=.9\columnwidth]{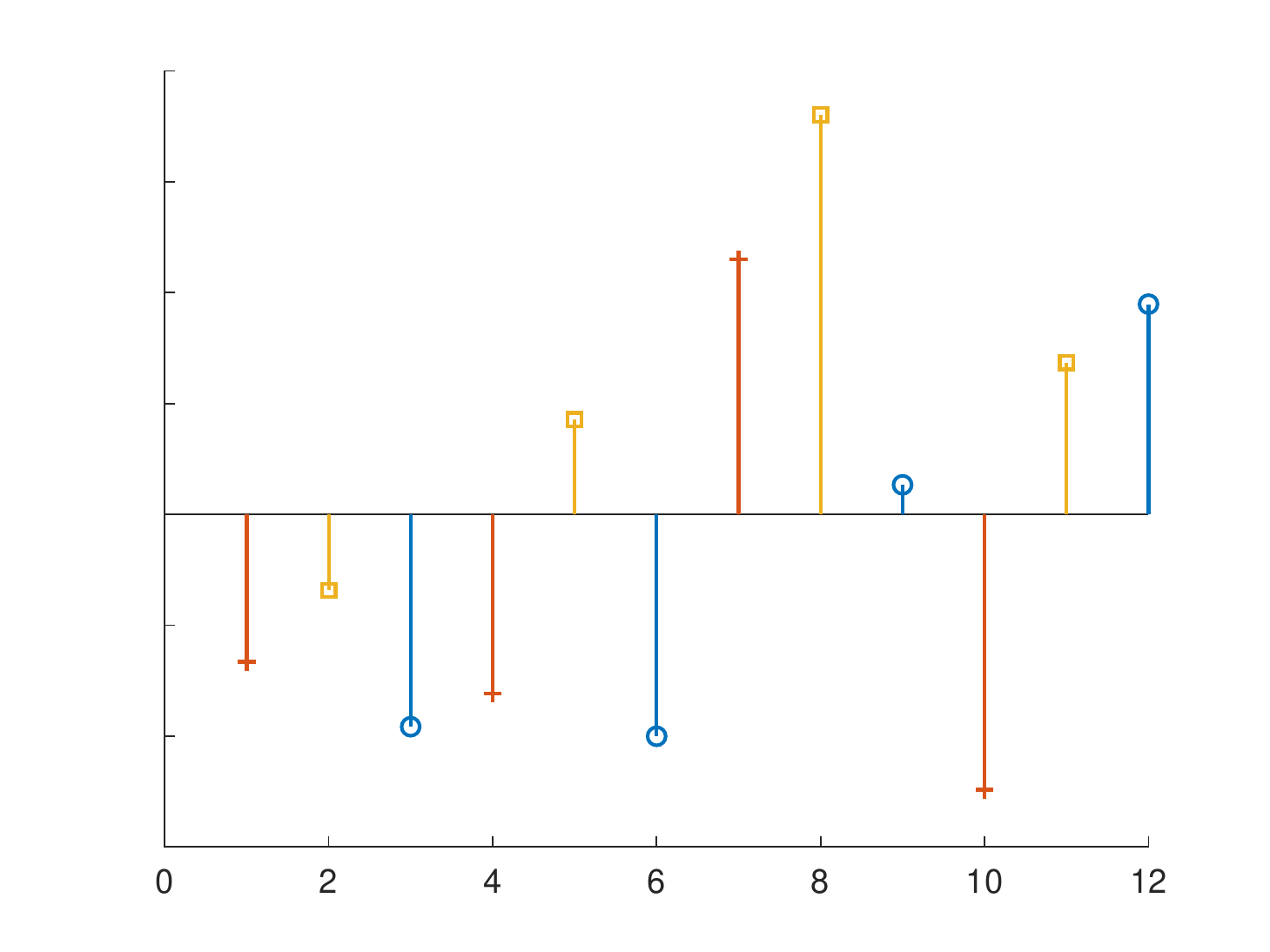}
	\caption{ Permutation + circular shifts of the sub-signals}
\end{subfigure}
\caption{\label{fig:group} An illustration of the orbit $G_{\Pi,L}x$; all four signals have the same likelihood function.  
	(a) A signal of length $M=12$ consists of $K=3$ sub-signals (drawn in different colors). (b) Permuting the sub-signal $(x_0,x_1,x_2)\mapsto (x_1,x_2,x_0)$. (c) Shifting the sub-signals {$(x_0,x_1,x_2)\mapsto (R_{-1}x_0,R_{-2}x_1,x_2)$}. (d) Permuting and shifting the sub-signals {$(x_0,x_1,x_2)\mapsto (R_{-2}x_1,x_2,R_{-1}x_0)$}.}
\end{figure}

Observations from the hMRA model~\eqref{eq:heter_mra} enable the recovery of 
$x_0,\ldots, x_{K-1}$ up to a circular shift of each sub-signal and a permutation across  signals. 
This can be seen by considering the 
marginalized likelihood of a single observation $y$: 
\begin{equation} \label{eq:likelihood}
p(y|x)=\frac{1}{M}\sum_{\ell=0}^{L-1}\sum_{k=0}^{K-1}
p(y|x,\ell,k) = \frac{1}{(2\pi\sigma^2)^{L/2}M}\sum_{\ell=0}^{L-1}\sum_{k=0}^{K-1}e^{-\frac{1}{2\sigma^2}\|y-R_{\ell}x_k\|_2^2}.
\end{equation}
Plainly, $p(y|x)$ is invariant under any permutation $\pi$  (overall $K!$  permutations)  $$(x_0,\ldots,x_{K-1}) \mapsto (x_{\pi(0)},\ldots,x_{\pi(K-1)}),$$ or circular shifts $\ell_0,\ldots,\ell_{K-1}$ (overall $L^K$  permutations) $$(x_0,\ldots,x_{K-1}) \mapsto (R_{\ell_0}x_0,\ldots,R_{\ell_{K-1}}x_{K-1}).$$ This set of  permutations, denoted by $G_{\Pi,L}$, includes  $K!L^{K}$ elements and constitutes a subgroup of the permutation group of $M$ elements. 
The orbit of $x$ under $G_{\Pi,L}$ is illustrated in Figure~\ref{fig:group}.

Importantly, previous works on hMRA aimed to retrieve  the orbit $G_{\Pi,L} x$.
In this work we further  wish to recover the high-resolution signal $x$ by ordering the sub-signal\rev{s} $x_0,\ldots,x_{K-1}$ properly, which is impossible based solely on the likelihood.  
Thus, we must impose some additional constraints on the signal. In particular, we show in Section~\ref{sec:merge_signals} that for almost any Gaussian prior there is a single element of  
 $G_{\Pi,L} x$ that achieves the maximum of the posterior distribution. 
 
\subsection{Identifying the orbit $G_{\Pi,L}x$} \label{sec:recover_orbit}

\subsubsection{The noiseless case}
\label{sec:noiseless}

In the absence of noise, if we have observed each one of the $K$ sub-signals $x_0,\ldots,x_{K-1}$, we can determine the orbit $G_{\Pi,L}x$ immediately by considering all of their circular shifts and permutations. Therefore, the only question is how many observations from~\eqref{eq:model} are required to see each sub-signal~$x_k$ at least once; this problem is known in the combinatorics literature as  the \emph{coupon collector's problem} \rev{(see, for instance, \cite{feller1968introduction})}.   In expectation, it is known that $KH_K$ observations are required to see all $K$ signals, where $H_K$ is the harmonic sum
\begin{equation*}
H_K = \frac{1}{1} + \frac{1}{2} + \cdots + \frac{1}{K}= \log K + \gamma + \epsilon_K,
\end{equation*}
where $\gamma\approx 0.57721 $ is the  Euler-Mascheroni constant and $\epsilon_K \sim 1/\rev{(2K)}$ for large $K$. If $K$ is large enough, the harmonic sum can be bounded by $H_K\leq C \log K$ for some small constant $C$. For example, $H_K\leq 2 \log K$ for any $K\geq 3$. Thus, we say that in expectation \rev{$N\approx K\log K=M/L\log(M/L)$} observations suffice to characterize the orbit $G_{\Pi,L}x$  from \rev{noiseless} observations.
Yet, even in the absence of noise, \rev{Theorem~\ref{thm:nonunique} suggests that} finding $x$ itself from $G_{\Pi,L}x$ is a non-trivial task, \rev{requiring additional assumptions; we address this in Section~\ref{sec:merge_signals}.}

\subsubsection{Auto-correlation analysis}
\label{sec:auto-correlations}
In the low SNR regime, we propose to estimate the orbit $G_{\Pi,L} x$ using the first three auto-correlations of the observations, or, equivalently, their Fourier counterparts: the mean, power spectrum, and bispectrum.  
Assuming $N/\sigma^6\to\infty$ and considering~\eqref{eq:heter_mra}, the  invariants of the data converge to the average of the invariants of the $K$ sub-signals, up to bias terms:
\begin{align} \label{eq:mix_invariants}
\frac{1}{N} \sum_{i=1}^{N}\hat{M}_1\left(y_i\right)& \overset{N\to\infty}{\to}  \frac{1}{K}\sum_{k=0}^{K-1}\hat{M}_1(x_k), \nonumber\\
\frac{1}{N} \sum_{i=1}^{N}  \hat{M}_2\left(y_i\right)[\ell] &\overset{N\to\infty}{\to} \frac{1}{K}\sum_{k=0}^{K-1}  \hat{M}_2(x_i)[\ell] + B_2, \\
\frac{1}{N} \sum_{i=1}^{N}\hat{M}_3\left(y_i\right)[\ell_1,\ell_2] &\overset{N\to\infty}{\to} \frac{1}{K}\sum_{k=0}^{K-1} \hat{M}_3(x_i) [\ell_1,\ell_2] + B_3[\ell_1,\ell_2], \nonumber
\end{align}
where $B_2 = \sigma^2L\mathbf{1}$, $B_3[\ell_1,\ell_2]=\bar{x} \sigma^2 L^2 D[\ell_1,\ell_2],$
$\mathbf{1}\in\mathbb{R}^L$ is a vector of ones, $\bar{x}$ is the average of $x$ 
$D\in\mathbb{R}^{L\times L}$ is a zero matrix except $D[0,0]=3$ and $D[i,0]=D[0,i]=D[i,i]=1$ for $i=1,\ldots,L-1$,
and $\ell,\ell_1,\ell_2=0,\ldots,L-1$.  We note that if $\sigma^2$ is known, one can easily remove the bias factors from the second- and third-order invariants.
As $N\to\infty$, the left-hand side equals the  right-hand side almost surely. 

The reason we require  $N/\sigma^6\to\infty$ is that  the third-order auto-correlation requires taking triple products of three noise terms (thus tripling the effective noise level), and thus for large $\sigma$ the number of observations needs to scale at least as $\sigma^6$ to keep the variance of the estimator under control; more precisely, only when $N/\sigma^6\to\infty$ one can estimate the invariants accurately.   
Therefore, if $N/\sigma^6\to\infty$, one can estimate the first three auto-correlations at any SNR levels. If $\sigma$ is fixed while $N\to\infty$, then one can estimate all auto-correlations at any SNR level. 
If $\sigma\to\infty$ and $N$ does not scale with $\sigma^6$, \rev{namely, $N/\sigma^6<C$ for some finite constant $C$, }then the third-order auto-correlation cannot be \rev{consistently} estimated from the observations. 

\subsubsection{Identifiability conditions for the orbit $G_{\Pi,L}x$ from the auto-correlations} \label{sec:identification_Gx}
As discussed in Section~\ref{sec:background},  it is well-known that the first three invariants determine a single \rev{generic} signal uniquely~\cite{sadler1992shift,kakarala2012bispectrum,bendory2017bispectrum}.  
Using tools from invariant theory and algebraic geometry, this result was recently extended to demixing of $K\geq 1$ invariants as in~\eqref{eq:mix_invariants}~\cite{bandeira2017estimation}.
The framework of~\cite{bandeira2017estimation} is based on  \rev{checking computationally the rank of the Hessian matrix of the map between a generic\footnote{{By the notion of generic signals, we mean that all  signals that are not recoverable in the regime defined by equation~\eqref{eq:Pl} satisfy a certain polynomial equation, and thus are} of measure zero.} signal and the invariants. If the rank is sufficiently high ({depending} on the algebraic structure of the problem), we say that the Hessian test is passed, implying that the orbit $G_{\Pi,L}x$ can be identified for generic~$x$. The Hessian test is executed on pairs of parameters $(K,L)$; if $(K,L)$ passes the test, it implies identifiability for all pairs $(K',L), K'\leq K$. In particular, it was
verified\footnote{We extended the range of parameters examined by the authors of~\cite{bandeira2017estimation}. We thank Dr. Joseph Kileel for his assistance to execute this computational verification.}  for all $L\leq 192$ that a set of   
generic signals is determined uniquely from~\eqref{eq:mix_invariants}, up to the symmetries that form the group $G_{\Pi,L}$, as long as 
\begin{equation} \label{eq:Pl}
K<\mathcal{P}(L) := \frac{L+3+\left\lfloor L/2\right\rfloor +  \left\lceil (L-1)(L-2)/6\right\rceil}{L+1}\approx L/6.
\end{equation}
This immediately implies that   the  number of required samples is bounded by
\begin{equation}
	K = \frac{M}{L} \lesssim L/6 \qquad \Rightarrow \qquad  L \gtrsim \sqrt{6M}.
\end{equation}
This bound is conjectured  to hold true for any pair  $(L,K)$ that satisfies~\eqref{eq:Pl}; see~\cite{bandeira2017estimation,boumal2018heterogeneous}.  We note that the bound of~\eqref{eq:Pl} 
is tight in the sense that it agrees with a simple upper bound based on parameter counting: on the one hand, the bispectrum of a generic signal contains $\mathcal{P}(L)\cdot L\approx L^2/6$ distinct entries (out of $L^2$ entries in total), and on the other hand   $K$ signals consists of $KL$ parameters.}



The following proposition and conjecture summarize the result:
\begin{prop} \label{prop1}
	Suppose that we acquire an average of the mean, power spectrum, and bispectrum of  $K$ signals as in the right-hand side~\eqref{eq:mix_invariants}. 
	 Then,  \rev{for  $L\leq 192$ and any $K$}  that satisfies~\eqref{eq:Pl}, one can identify the orbit $G_{\Pi,L} x$ for generic $x$. 
\end{prop}	
\begin{conj} \label{conj1}
	Suppose that we acquire an average of the mean, power spectrum, and bispectrum of  $K$ signals as in the right-hand side of~\eqref{eq:mix_invariants}. Then,  for any pair $(K,L)$ that satisfies~\eqref{eq:Pl},
	one can identify the orbit $G_{\Pi,L} x$ for generic $x$.
\end{conj}	

\subsubsection{A note on computational considerations} \label{sec:computational_considerations}
It is important to note that \rev{Proposition~\ref{prop1}} does not claim that the bound~\eqref{eq:Pl} can be achieved using a computationally efficient (e.g., polynomial time) algorithm.  
In the context of the hMRA model, numerical evidence suggests that for i.i.d.\ standard Gaussian signals, one can estimate the 
orbit  $G_{\Pi,L}x$ from a mix of bispectra using  non-convex least squares only in the regime $K\leq \sqrt{L}$---substantially below the identifiability regime $K\lesssim L/6$~\cite{boumal2018heterogeneous}.
Recently, it was proven that
i.i.d.\ standard Gaussian signals can be disentangled,  with high probability,   using a sum-of-squares algorithm  as long as $K\leq \sqrt{L}/\text{polylog}(L)$~\cite{weinthesis}.
In~\cite{bandeira2017estimation,boumal2018heterogeneous}, it was conjectured that the $\sqrt{L}$ bound reflects a fundamental statistical-computational gap\rev{---namely, while {it is statistically possible} to recover approximately $L/6$ signals, any efficient  (polynomial-time) algorithm can estimate at most $\sqrt{L}$ signals.}

If indeed one can recover only up to $\sqrt{L}$ signals efficiently from~\eqref{eq:mix_invariants}, it  implies that {the orbit $G_{\Pi,L}x$  can be estimated efficiently in the regime \rev{$K=\frac{M}{L}\leq\sqrt{L} \Rightarrow M \leq L^{3/2}$} for i.i.d.\ Gaussian entries. 
{Having said that, in contrast to the  hMRA model, the goal of the super-resolution problem is not to recover the orbit $G_{\Pi,L}x$, but rather the signal $x$ itself; the latter task seems to be a  significantly  more challenging computational problem. 
\rev{In addition, the main interest of this work is in smooth signals (e.g., signals with decaying power spectrum). 
For such signals, the achievable performance of hMRA deteriorates~\cite{bandeira2017optimal}, and even recovering  $\sqrt{L}$ signals seem to be unreachable.} 	
Indeed, numerical experiments in Section~\ref{sec:numerics} suggest that recovery is not attainable even in the regime $M \approx L^{3/2}$---at least not with the EM algorithm.}

	
\subsection{Identifying a unique high-resolution signal from the orbit $G_{\Pi,L} x$}\label{sec:merge_signals}
Until now, we have shown that one can identify the orbit $G_{\Pi,L}x$ for generic $x$ if $L\gtrsim \sqrt{6M}$\rev{, $L<192,$} and  the first three auto-correlations can be estimated from the observations.  Next, we wish to show how to determine a single signal out of $G_{\Pi,L}x$. 

Recall that the posterior distribution  $p(x|y_1,\ldots,y_N)$ is proportional to the likelihood function $p(y_1,\ldots,y_N|x)$ times a prior on the signal $p(x)$. According to~\eqref{eq:likelihood}, the likelihood is constant over the orbit $G_{\Pi,L}x$. \rev{Consequently, we choose the signal in the orbit that best fits the prior.} Importantly, this part is independent of the observations.

Many priors can be used. In this section, we focus on  Gaussian signals with zero mean and covariance $\Sigma$, that is, $p(x)=\frac{1}{\sqrt{(2\pi)^{M}|\Sigma|}}e^{-\frac{1}{2}x^T\Sigma^{-1}x}$. \rev{Priors of this form are ubiquitous in signal processing, and have been considered in previous works, such as \cite{robinson2006statistical, robinson2009optimal, woods2005stochastic}.} In particular, we wish to show that among all signals in $G_{\Pi,L}x$, there is a single signal that maximizes $p(x)$, or, equivalently, minimizes $x^T\Sigma^{-1} x$ for a positive-definite matrix, $\Sigma^{-1}$. The next lemma shows that permuting a signal usually changes this quadratic form. 
To this end, we define the  $\Sigma^{-1}$ norm  by $\norm{z}_{\Sigma^{-1}}^2 := z^T \Sigma^{-1} z$ for a positive-definite matrix $\Sigma^{-1}$. 
\begin{lemma} \label{lemma:unique_perm}
Let \rev{$\Omega \subset \mathbb{R}^n$ and denote by} $R_1$ and $R_2$ two permutation matrices and their ratio by $R=R_2 R_1^T$. Then the set 
\[  \mathcal{Z}_{R_1,R_2} =  \left\lbrace z \in \Omega \mid \norm{R_1z}_{\Sigma^{-1}}^2 = \norm{R_2 z}_{\Sigma^{-1}}^2 \right\rbrace  ,\] 
is a subset of $\Omega$ of measure zero if and only if ${\Sigma^{-1}}$ and $R$ do not commute.
\end{lemma}
\begin{proof}
Let $A := R_1^T {\Sigma^{-1}} R_1 - R_2^T {\Sigma^{-1}} R_2$. Observe that $A = 0$ if and only if $R$ and {$\Sigma^{-1}$} commute. The condition $\norm{R_1z}_{\Sigma^{-1}}^2 = \norm{R_2 z}_{\Sigma^{-1}}^2$ is equivalent to $ z^T A z = 0$. Thus, if $A=0$ then  $\mathcal{Z}_{R_1,R_2}=\Omega$. Otherwise, $A\neq0$ \rev{and symmetric, i.e., there exists a basis $Q$ of orthogonal eigenvectors such that $(Qy)^T A Qy = \sum_{i=1}^n \lambda_i y^2_i$, with $\lambda_1 \neq 0$. Therefore, the condition $(Qy)^T A Qy = 0$ means that  $y_1^2 = \sum_{i=2}^n \frac{\lambda_i}{-\lambda_1} y^2_i$. Denote the indicator functions $\chi_Z$ and $\chi_Y$ for the two sets $\left\lbrace z \in \Omega \mid z^TAz = 0 \right\rbrace $ and $\left\lbrace y \in \mathbb{R}^{n} \mid y^T (Q^TAQ) y = 0\right\rbrace$, respectively. \rev{Since orthogonal transformations preserve integrals} and $\Omega \subset \mathbb{R}^n$, we have
\[ \abs{\mathcal{Z}_{R_1,R_2} } = \int_\Omega \chi_{Z}(z) dz  \le \int_{\mathbb{R}^n} \chi_{Y} (y) dy = \int_{\mathbb{R}^{n-1}} \left( \int_{\mathbb{R}} \chi_{Y} (y_1,y_2,\ldots,y_n ) dy_1 \right) dy_2 \cdots dy_n .  \]
Recall that $(Qy)^T A Qy = 0$ implies that $y_1 = \pm \left( \sum_{i=2}^n \frac{\lambda_i}{-\lambda_1} y^2_i\right)^{\frac{1}{2}}$. Thus, for any fixed $y_2,\ldots,y_n$, the indicator function $\chi_{Y}$ \rev{is nonzero on only two points, implying}
$\int_{\mathbb{R}} \chi_{Y} (y_1,y_2,\ldots,y_n ) dy_1  = 0$. Consequently, $\abs{\mathcal{Z}_{R_1,R_2} } = 0$.}

\end{proof}

Any group of permutations over \rev{a finite set is finite}, and thus there are finitely many pairs $R_1$ and $R_2$ of permutation matrices. Consequently, the set of signals for which  $\norm{R_1z}_{\Sigma^{-1}}^2 = \norm{R_2 z}_{\Sigma^{-1}}^2$  for some pairs of permutations is also of measure zero.
\begin{cor} \label{cor:prior}
	Assume the conditions of Lemma~\ref{lemma:unique_perm} \rev{are met, that is, ${\Sigma^{-1}}$ and {$R_2 R_1^T$} do not commute for all pairs $R_1\neq R_2$ in the permutation group}. Then \rev{for almost every $x$} in $\Omega$ the minimum of the quadratic form $y^T\Sigma^{-1} y$ is unique among all signals $y$ in $G_{\Pi,L}x$. 
\end{cor}

The case when $\Sigma^{-1}$ is a circulant matrix is of particular importance. Such a prior, reflecting a prior on the signal's power spectrum, is popular in many signal processing tasks, as well as in cryo-EM; see for instance~\cite{scheres2012relion}. 
	 In this case, both the prior and the likelihood function (and thus the posterior) are 
shift-invariant, and therefore any signal is indistinguishable from its cyclic shifts. To account for this symmetry, we derive a distinct uniqueness result---up to a circular shift---for circulant matrices.

\begin{prop} \label{prop:symmetry_in_circulant_sigma}
Assume $\Sigma^{-1}$ is a circulant, positive-definite matrix of size $n>2$. Then:
\begin{enumerate}
\item
The $\Sigma^{-1}$ norm $\norm{ \cdot }_{\Sigma^{-1}}$ is invariant under cyclic shifts. Consequently, we may consider the quadratic form $y^T \Sigma^{-1} y$ as a function over equivalence classes in $G_{\Pi,L} x$, where two vectors are equivalent if one is a cyclic shift of the other.

\item
If all eigenvalues of $\Sigma^{-1}$ are distinct, then for almost every signal $x$ the minimum of the quadratic form $y^T\Sigma^{-1} y$ is unique over the equivalence classes of $G_{\Pi,L} x$.
\end{enumerate}
\end{prop}
The proposition is proved in Appendix~\ref{app:R_invariant}.

\section{An expectation-maximization algorithm }
\label{sec:EM}

Our theoretical study is based on invariant features. Conceptually, it suggests a two-stage procedure: it begins by identifying the orbit $G_{\Pi,L}x$, and 
then choosing a  unique signal according to the prior.
While identifying the orbit $G_{\Pi,L}x$ efficiently using bispectrum demixing is possible~\cite{boumal2018heterogeneous,weinthesis}, it is unclear how to devise a tractable algorithm for the second step.
As an alternative, we formulate an EM algorithm, described below, which aims to achieve the maximum of the posterior distribution by maximizing the likelihood function and the prior simultaneously~\cite{dempster1977maximum}. \rev{EM  is known to work quite well in many practical scenarios; see for instance its application to cryo-EM experimental datasets~\cite{scheres2012relion,punjani2017cryosparc}.
 In what follows, we formulate EM for MRA with a general linear operator
  \begin{equation}
 	y = T R_s x + \varepsilon, 
 \end{equation}
 where $s$ is drawn from a uniform distribution on  a discrete grid $S_M$ with $M$ points, and $T$  is a general linear operator (not necessarily a sampling matrix).
 For the special case of super-resolution, a similar algorithm was already derived by~\cite{woods2005stochastic}.}

EM is a common framework to compute the  maximum aposteriori estimator (MAP) 
Hereafter, we formulate  EM for the general model

Given a set of $N$ independent observations $y_1,\ldots,y_N$, the \rev{log-}posterior distribution \rev{$\log p(x|y_1,\ldots,y_N)$} is 
proportional to \rev{$\log p(y_1,\ldots,y_N|x)p(x)$}, where  
\rev{
	\begin{equation} \label{eq:likelihood_em}
		\log p(y_1,\ldots,y_N|x)= \sum_{i=1}^N\log\sum_{{s_{\ell}}\in S_M}e^{-\frac{1}{2\sigma^2}\|y_i-TR_{s_\ell} x\|_2^2} + \text{constant},
	\end{equation}
}
is the \rev{log-}likelihood function, and $p(x)$ is a prior on the signal. 
We assume that the signal is drawn  from a Gaussian prior with zero mean and known covariance $\Sigma$ so that ${p(x)\sim\mathcal{N}(0,\Sigma)}$. 
EM aims to maximize the posterior iteratively, where each iteration consists of two steps.
The first step, called  E-step, computes the expected value of  the likelihood $x$ (note, not the marginalized likelihood) with respect to the circular shifts (i.e., the nuisance variables), given the current estimate of the signal $x_t$ and the data $y_1,\ldots,y_N$:
\rev{
\begin{equation}
\begin{split}
Q(x|x_t) &= \E_{s_1,\ldots,s_N|y_1,\ldots,y_N,x_t}\left\{\log p(y_1,\ldots,y_N,s_1,\ldots,s_N|x) +\log p(x) \right\} \\
& = -\frac{1}{2\sigma^2}\E_{s_1,\ldots,s_N|y_1,\ldots,y_N,x_t}\left\{\sum_{i=1}^{N}\|y_i - TR_{s_i}x\| \right\} - \frac{1}{2}x^T\Sigma^{-1}x + \text{constant}\\
& = -\frac{1}{2\sigma^2}\sum_{i=1}^{N}\sum_{s_\ell\in S_M}w_{i,\ell}\|y_i - TR_{s_\ell}x\| - \frac{1}{2}x^T\Sigma^{-1}x + \text{constant},
\end{split}
\end{equation} 
}
where 
\begin{equation} \label{eq:em_weights}
w_{i,\ell} =  \frac{e^{\frac{-1}{2\sigma^2}\|y_i-TR_{s_\ell} x_t\|^2 }}{\sum_{{s_\ell}\in S_M}e^{\frac{-1}{2\sigma^2}\|y_i-TR_{s_\ell} x_t\|^2 }}.
\end{equation}
{The}  second step, called M-step, maximizes $Q$ with respect to $x$. 
The solution is obtained by solving the linear system of equations
\begin{equation} \label{eq:em_linear_system}
Ax = b,
\end{equation}
where 
\begin{align}
A :=&  \Sigma^{-1} + \frac{1}{\sigma^2}\sum_{i,\ell}\omega_{i,\ell} (R_{\ell}^{-1}T^TTR_{\ell}),\\ 
b :=&   \frac{1}{\sigma^2} \sum_{i,\ell}\omega_{i,\ell}R_\ell^{-1}T^Ty_i.
\end{align}
The EM algorithm iterates between computing the weights~\eqref{eq:em_weights} and solving the linear system~\eqref{eq:em_linear_system} until convergences. 
In our implementation, the algorithm halts when the relative difference between the posterior of two consecutive iterations falls below $10^{-5}$. 
In general, EM {does} not converge to the global maximum of the posterior distribution; however, 
each iteration is guaranteed not to decrease the posterior~\cite{dempster1977maximum}. 
In addition, several recent works derived intriguing theoretical results for EM under specific statistical models.  See for example~\cite{balakrishnan2017statistical,daskalakis2016ten}, and in particular~\cite{fan2020likelihood} that analyzes EM for the MRA model. 

In the special case in which the linear operator $T$ is the sampling operator~\eqref{eq:model}, computing the weights and constructing $A$ and $b$ reduces to computing a set of correlations; this can  be executed  efficiently using FFT. \rev{For EM implementation when  a blurring kernel is included, see~\cite{woods2005stochastic}.}

\rev{ \paragraph{On the connection between  maximum likelihood estimation and the invariant approach.} 
A recent paper by Katsevich and Bandeira~\cite{katsevich2020likelihood}  studies Gaussian mixture models, for which heterogeneous MRA is a special case, in the parametric setup considered in this work: $N\to\infty$, SNR$\to 0$, and fixed $M$. In particular, they show that  log-likelihood maximization is equivalent to an asymptotic series of  successively higher moment matching problems. In this sense, a method based on the bispectrum (a third-order moment), as we use for the analysis, can be thought of as a third-order approximation of the likelihood function.
}

\section{Numerical results} \label{sec:numerics}

We conducted three experiments to examine the performance of the EM algorithm. The first experiment demonstrates resolving  two adjacent peaks from low-resolution observations. The next two experiments \rev{study} the performance of the algorithm as a function of noise level and the number of samples. The code for all experiments and the EM algorithm is publicly available at \url{https://github.com/TamirBendory/MRA-SR}. 

In all experiments, \rev{we set a  prior on the  signal's power spectrum, and  therefore,} to account for the circular shift symmetry, the relative recovery error is defined as 
\begin{equation}
\mbox{relative error} = \min_{\ell\in\mathbb{Z}^M}\frac{\|R_\ell x_\text{est} - x\|}{\|x\|},
\end{equation}
where $x\in\R^M$ is the {underlying} signal, and $x_\text{est}$ is the output of the EM algorithm.
The SNR is defined as 
$\mbox{SNR} = \|x\|^2/(M\sigma^2)$. 
The EM iterations terminate when either a maximal number of iterations is reached, or the relative absolute difference of the posterior between two consecutive iterations drop below a tolerance parameter. 
In all experiments,  the maximal number of EM iterations was set to be 100 iterations, and   the tolerance parameters was~$10^{-5}$.   
The EM may be initialized from multiple random points and thus produce different  estimators. 
Among those estimators, we choose the one with the largest 
posterior. The posterior is computed at each {EM} iteration. 

\paragraph{Experiment 1.} 
The signal in this experiment is of length $M=120$ with bandwidth  (the largest non-zero frequency, see~\eqref{eq:continuous_signal_sampled})  $B=15$. 
 We generated $N=10^4$ observations by shifting the signal  and sampling it at $L=15$ points, corresponding to half of the Nyquist {sampling} rate.
Then, an i.i.d.\ Gaussian noise was added, corresponding to SNR = 1. 
We ran the EM algorithm from five random initial points; each trial required 13 to 17 iterations to converge. 
Since the signal (only in this experiment) is bandlimited, in each iteration the current estimate is projected onto the low $B=15$ frequencies. 
The target and estimated signals are presented in Figure~\ref{fig:recovery_example}; 
the relative recovery error is 0.0614. 
 Figure~\ref{fig:recovery_example_per_freq} displays the relative error per frequency. As can be seen, the relative error of frequencies above the Nyquist rate is still quite low, indicating that the EM algorithm {resolves} high frequencies accurately.

\begin{figure}[ht]
	\centering
	\includegraphics[width=.6\linewidth]{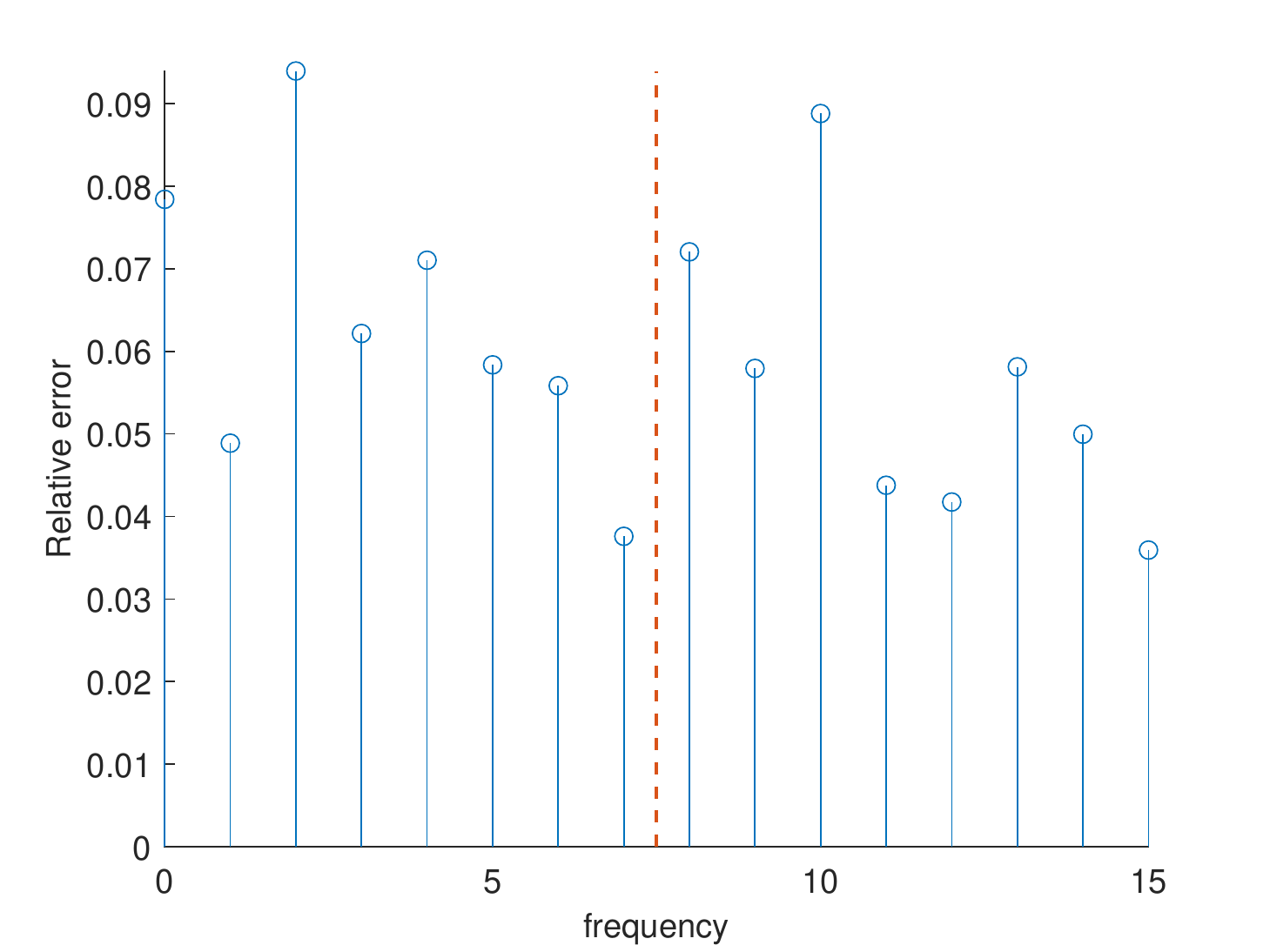}
	\caption{\label{fig:recovery_example_per_freq} 
		Recovery error per frequency of the \rev{experiment} presented in Figure~\ref{fig:recovery_example}. The figure indicates that the EM algorithm succeeds to resolve frequencies beyond the largest frequency {determined} by the Nyquist sampling rate $L/2$ (vertical red line).
	}
\end{figure}

\paragraph{Experiment 2.}
Figure~\ref{fig:snr} presents the error curve as a function of the SNR, in the high and low SNR regimes.
For each SNR value, 50 trials were conducted and we present the median error.
A signal of length $M=64$ was drawn from a Gaussian distribution with zero mean, and a circulant covariance matrix, corresponding to power spectrum decaying linearly, that is, as $1/f$. Following the circular shift, each observation was sampled at $L=32$ equally-spaced points---corresponding to half of the Nyquist {sampling} rate. 

Figure~\ref{Fig3:high} presents the relative error as a function of the SNR, for 30 SNR values sampled uniformly on a logarithmic scale between $10^{0.2}$ to $10^{2}$, and $N=10^2$ observations; this reflects the high SNR regime. 
Unfortunately,  the EM seems to suffer from a flaw: it should be initialized from many points (among them we choose the one with the largest posterior value) in order to result in a consistent recovery. In this experiment, we initialized the algorithm from 1000 points; the computational load is still quite cheap since each trial requires only a few iterations (around 5). Yet, it suggests that EM may not be the optimal computational scheme in the high SNR regime.  The slope of the error curve is approximately \rev{-1/2}. \rev{Since the SNR is proportional to $1/\sigma^2$, this  indicates that the error  scales as} $\sigma$---the optimal estimation rate even if the circular shifts were known. 

Figure~\ref{Fig3:low} shows a similar experiment for SNR values ranging between  $10^{-0.6}$ to 1 and $N=10^5$ observations. 
In this low SNR regime, the EM algorithm seems to be more consistent, and thus we initialized it from merely 20 random points.
In this regime,  the slope  of the error curve becomes steeper and the error slope is smaller than $-1$, implying that \rev{the error scales faster than~$\sigma^4$}.
\rev{This indicates, in line with previous works on MRA (with finite $M$), that the estimation rate in the high and low SNR regimes is drastically different~\cite{bandeira2017optimal,perry2017sample,abbe2018multireference,abbe2018estimation}.}
In fact, our analysis predicts that as SNR$\to 0$, the slope of the error curve would tend to $\mbox{SNR}^{-3/2}$; see Section~\ref{sec:analysis}.

\begin{figure}
	\begin{subfigure}[ht]{0.48\columnwidth}
		\centering
		\includegraphics[width=\columnwidth]{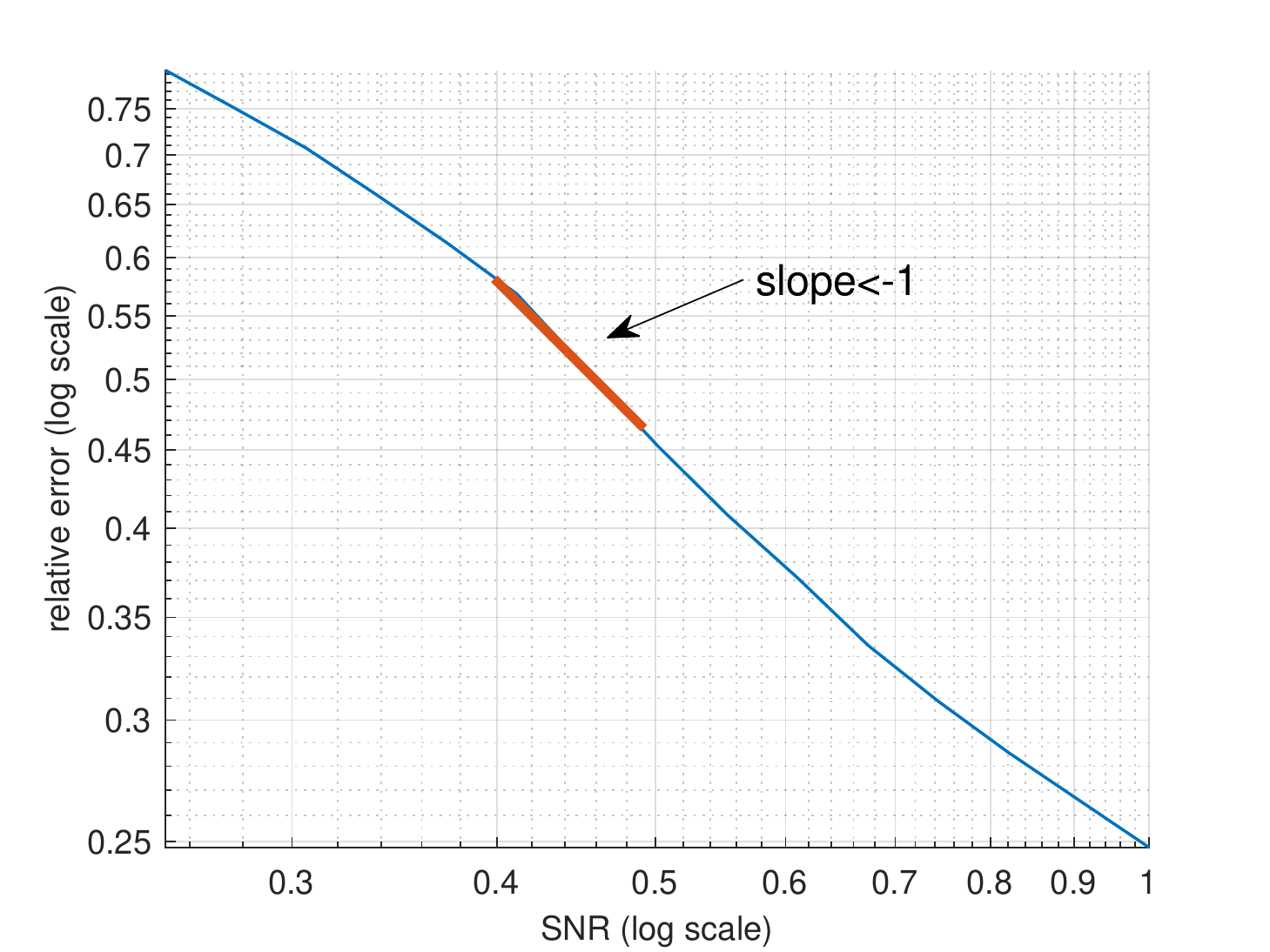}
		\caption{\label{Fig3:low} Low SNR}
	\end{subfigure}
	\hfill
	\begin{subfigure}[ht]{0.48\columnwidth}
	\centering
	\includegraphics[width=\columnwidth]{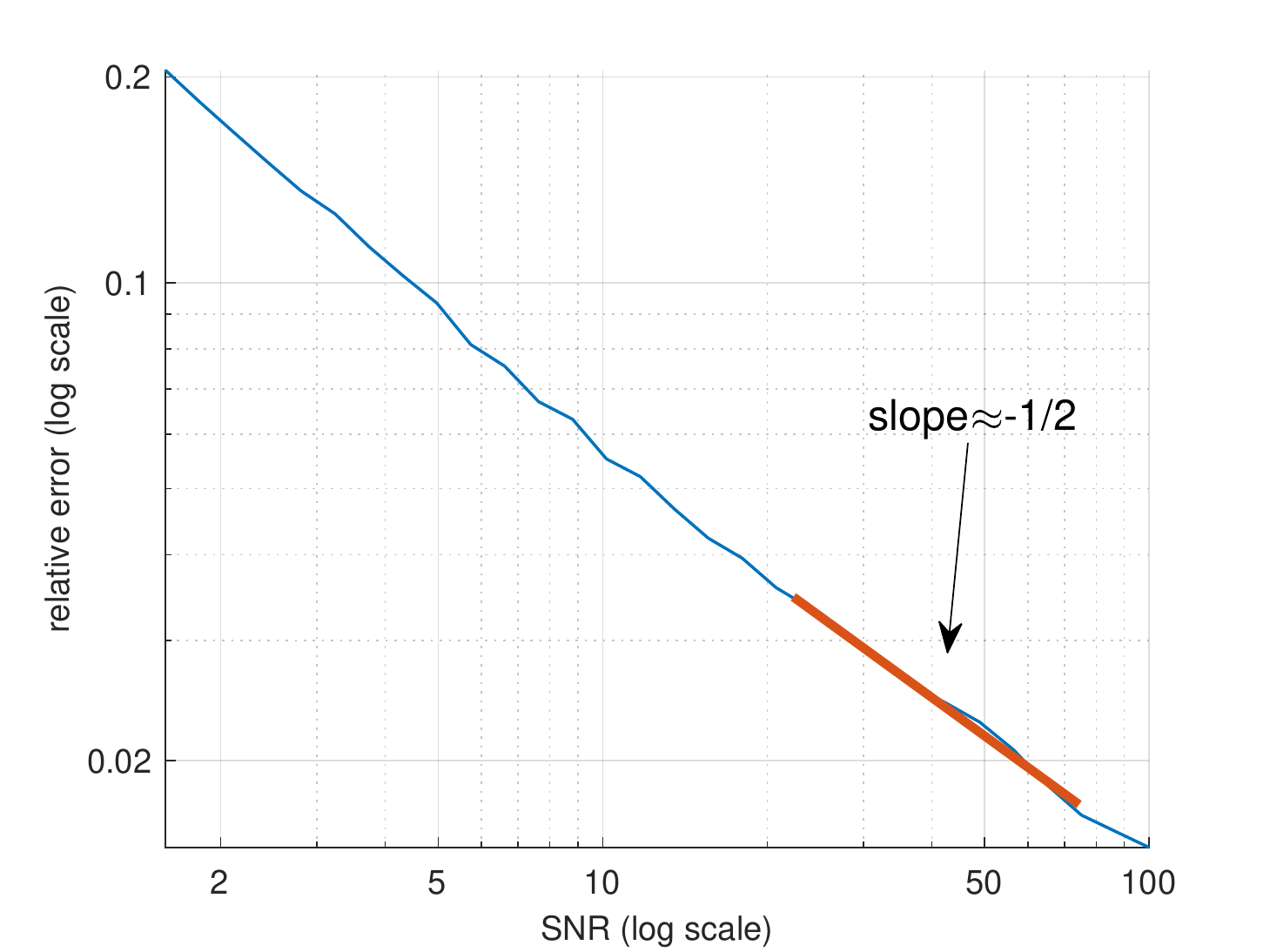}
	\caption{\label{Fig3:high} High SNR}
\end{subfigure}

	\caption{\label{fig:snr} Relative estimation error as a function of the SNR. In the high SNR regime, the relative error scales as SNR$^{-1/2}$, which is the same estimation rate as if there were no shifts (namely, the estimation rate of averaging independent Gaussian variables). In the low SNR regime, the error decays faster than SNR$^{-1}$, demonstrating a sharp transition from the high SNR regime.}
\end{figure}

\paragraph{Experiment 3.}
Figure~\ref{fig:error_L} examines the recovery error for  different values of~$M$ and $L$. For each $M$, we chose values of $L$ so that $M/L$ is an integer.
The signals were generated as in Experiment 2 with SNR = 5 \rev{and $N=1000$}, and the EM was initialized from 50 random locations {in each trial}.
For each pair (M,L), the mean error over 50 trials was recorded. 
The red vertical dashed  line indicates $L=M^{2/3}$; this is the conjectured computational recovery limit for hMRA, namely, for recovering the orbit~$G_{\Pi,L}x$ (see discussion in Section~\ref{sec:computational_considerations}). 
Notwithstanding, we get relatively small recovery error only for much larger values of $L$, {suggesting} that the super-resolution problem is {computationally} more challenging than hMRA.
In particular, our theoretical analysis is split into two stages: recovering the orbit $G_{\Pi,L}x$, and recovering $x$ from the orbit; the latter depends only on the prior, and not on the data. 
In contrary, the EM algorithm aims to implicitly carry out both stages simultaneously. 
We believe that the second stage\rev{, together with the smoothness of the signals (see Section~\ref{sec:computational_considerations}),} is the reason  the performance of EM for super-resolution is inferior to what was demonstrated in previous MRA setups~\cite{bendory2017bispectrum,abbe2018multireference,ma2018heterogeneous}. 

\begin{figure}
	\begin{subfigure}[ht]{0.32\columnwidth}
		\centering
		\includegraphics[width=1.1\columnwidth]{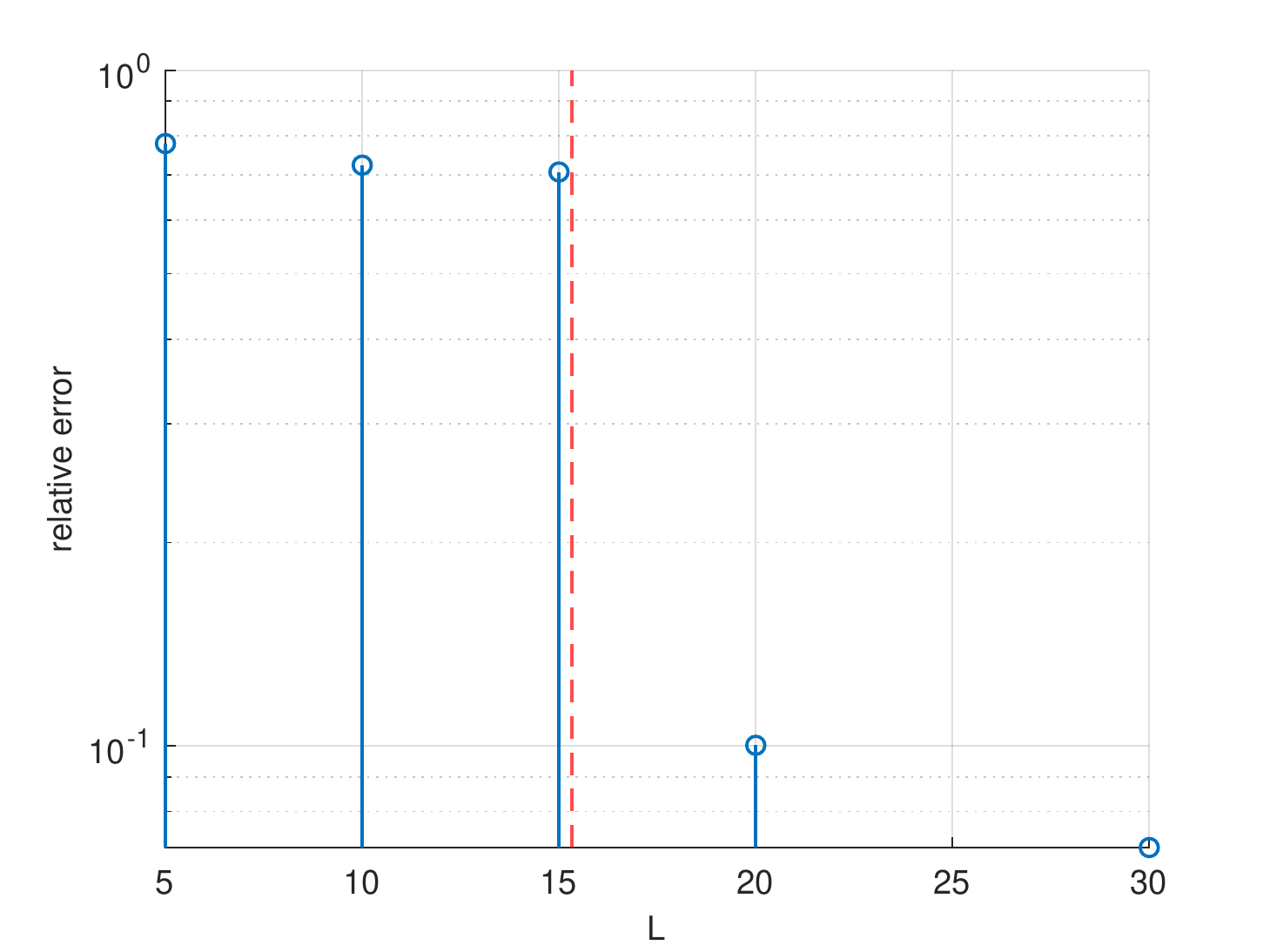}
		\caption{$M=60$}
	\end{subfigure}
	\hfill
	\begin{subfigure}[ht]{0.32\columnwidth}
		\centering
		\includegraphics[width=1.1\columnwidth]{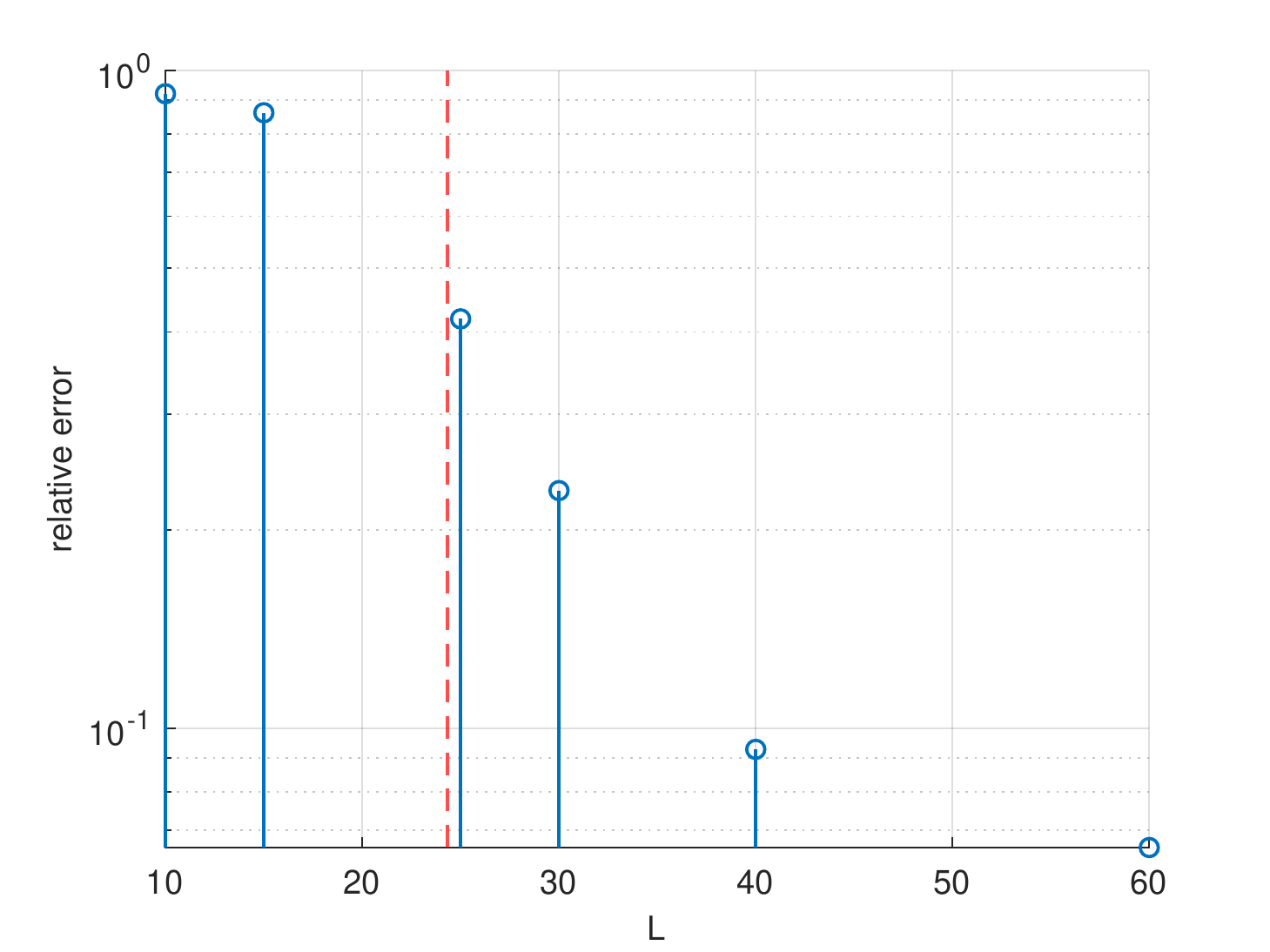}
		\caption{$M=120$}
	\end{subfigure}
	\hfill
	\begin{subfigure}[ht]{0.32\columnwidth}
		\centering
		\includegraphics[width=1.1\columnwidth]{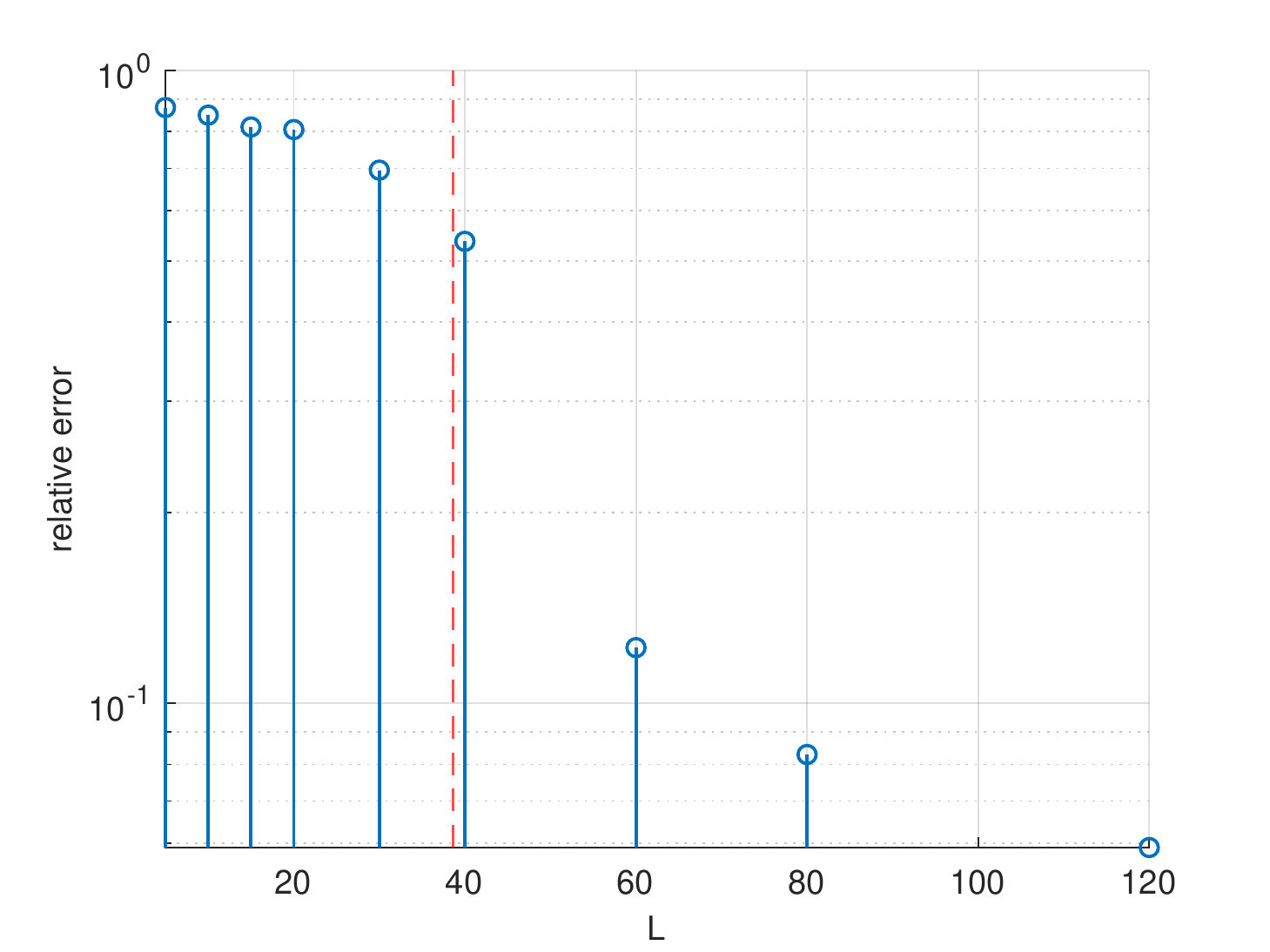}
		\caption{$M=240$}
	\end{subfigure}
	\caption{\label{fig:error_L} Relative recovery error as a function of $L$ for different values of $M$. The red dashed line indicates $L=M^{2/3}$. The results suggest that the super-resolution problem is significantly harder  than hMRA.}
\end{figure}

\section{Discussion} \label{sec:future_work}

\paragraph{Super-resolution limits.}
This work analyzes the super-resolution from multiple observations problem in a  noisy environment using the third-order auto-correlation.
To use higher-order auto-correlations, more observations should be collected: the number of observations needs to scale  {as} $\sigma^{2q}$ to estimate the $q$-th {order} auto-correlation accurately.
{The} $q$-th auto-correlation provides $O(L^{q-1})$ polynomial equations of the sought signals.
Based on our analysis and the reduction of the super-resolution problem to the hMRA model~\eqref{eq:heter_mra}, we expect that the $q\geq3$ auto-correlation would identify $M=O(L^{q-1})$ grid points. Such a  result will follow directly from a generalization of~\cite{bandeira2017estimation} to higher-order auto-correlations. This leads us to the following conjecture:
\begin{conj}
	Suppose that $N$ observations from~\eqref{eq:model} are collected and each observation is sampled at $L$ equally-spaced locations. 
	Then, in the low SNR regime $\sigma\to\infty$, if $N/\sigma^{2q}\to\infty$ for some $q\geq 3$, one can identify up to $M= O(L^{q-1})$ grid points. In other words, only $L=O(M^{1/{(q-1)}})$ samples per observation suffice for signal identification. In particular, for $N\to\infty$ and any fixed noise level (that might be arbitrarily high), there is no theoretical  limit on the achievable resolution.
\end{conj}

\paragraph{Continuous super-resolution.}
A natural generalization of the model considered in this work is the following. Let $x: S^1\rightarrow\R$ be a band-limited signal on the circle~\eqref{eq:continuous_signal}, and let $R_\theta$ denote a  rotation, that is, $(R_\theta x)(t)= x(t-\theta)$, where  $\theta$ is  distributed uniformly on the circle. Together with i.i.d.\ Gaussian noise $\varepsilon\in \R^L$,  the data generative model reads
\begin{equation} \label{eq:model_continuous}
y = P(R_\theta x) + \varepsilon, \qquad  \theta\sim \text{Uniform}[0,1),\qquad \varepsilon\sim \mathcal{N}(0,\sigma^2 I),
\end{equation}
where  $P$ denotes a sampling operator that collects $L$ equally-spaced point-wise samples. 
The goal is to estimate $x$ from $N$ observations sampled from~\eqref{eq:model_continuous}. This setup is interesting in the sub-Nyquist regime, where $P$ samples $x$  below its Nyquist  sampling rate. While this model shares many similarities with~\eqref{eq:model}, it poses some additional challenges that are beyond the scope of this work; we intend to address them in a follow-up work. 

\paragraph{\rev{Super-resolution of images.}} While this paper deals with 1-D signals, the methodology can be extended to higher-dimensions. 
For example, an interesting MRA setup that was studied in~\cite{ma2018heterogeneous} considers rotating 2-D ``bandlimited'' images. 
Specifically, suppose that  an image $X$ belongs to the vector space of images that can  expanded by finitely many coefficients in a steerable basis (such as Fourier-Bessel~\cite{zhao2013fourier} or prolate spheroidal wave functions~\cite{landa2017approximation}):
\begin{equation}
X(r,\phi) = \sum_{k=1}^{K_{max}}\sum_{q=1}^{Q_{max}} a_{k,q} u_{k,q}(r,\phi),
\end{equation}
where  $(r,\phi)$ are polar coordinates, $a_{k,q}$ are the expansion coefficients, and  $u_{k,q}(r,\phi)$ are the basis functions of the steerable basis.
The images are acted upon by unknown elements of the group of in-plane rotations $SO(2)$.
The \rev{steerability} property implies that rotating an image by an angle $\alpha$ amounts to multiplying the expansion coefficients by $e^{\I k \alpha}$:
\begin{equation}
X(r,\phi-\alpha) = \sum_{k=1}^{K_{max}}\sum_{q=1}^{Q_{max}} a_{k,q}e^{-\I k\alpha} u_{k,q}(r,\phi).
\end{equation}
Accordingly, it is easy to see that the triple products $a_{k_1,q_1}a_{k_2,q_2}a_{-k_1-k_2,q_3}$ are invariant under rotations---these products form the bispectrum~\cite{zhao2013fourier,ma2018heterogeneous}. This in turn implies that for an image expanded by $M$ coefficients, there are $O(M^{5/2})$ bispectrum entries. 
In this case, our framework suggests that, perhaps,  one can identify an image from sufficiently many observations with merely $L=O(M^{-2/5})$ samples per observation. 

\rev{\paragraph{Super-resolution in high dimensions.}
	This work studies the finite-dimensional regime (finite~$M$) in which invariant features achieve the optimal estimation rate as SNR $\to0$. A recent work~\cite{romanov2020multi} uncovered that this is not the case in the high-dimensional regime $M\to\infty.$ In particular, it was shown that the parameter that controls the ``hardness'' of the model is $\alpha =M/(\log M \sigma^2) $: when $\alpha<2$ the samples complexity of the problem rapidly increases, whereas for $\alpha>2$ the effect of the unknown shifts is minor. 
	Nevertheless, the authors of~\cite{romanov2020multi} only investigated  random sub-sampling operators that do not include the super-resolution setup~\eqref{eq:model}.
	The high-dimensional regime is of particular interest since it seems  as a good model for  modern cryo-EM datasets, where the dimensionality and the number of samples are of the same order of a few millions. In fact, high-dimensional statistical analysis has been
	already proven effective for cryo-EM data processing. For example, a covariance estimation
	technique based on high-dimensional analysis (the so-called spiked model) has significantly improved
	image denoising~\cite{bhamre2016denoising}.		
}

\paragraph{Cryo-EM and XFEL.}
The main motivation of this work arises from cryo-EM and XFEL. 
The measurements in these applications  (under simplifying assumptions, see for example~\cite{bendory2019single})
agree with the general MRA model~\eqref{eq:model}, where $g\in SO(3)$ (the group of 3-D rotations),  the 3-D Fourier transform of the signal $x$ is assumed to be bounded in a ball (``bandlimited'' volume), and the linear operator $T$ collects samples of  the 2-D tomographic projection of the rotated volume. The question then would be whether the maximal resolution of a 3-D reconstruction algorithm can surpass the resolution dictated by the detectors acquiring the data---that is, the resolution of the 2-D tomographic projection images. A recent proof \rev{of} concept (on simulated  data) promises an affirmative answer~\cite{chen2018single}. Extending our analysis to this case requires sophisticated tools 
 and we leave it for a future research.

\section*{Acknowledgment}
The authors are grateful to Joseph Kileel and Dan Edidin for  insightful discussions about algebraic geometry. 
\rev{The authors also thank the anonymous reviewers for their
	valuable comments and suggestions.}
A.S. and W.L. were partially supported by NSF BIGDATA award IIS-1837992. 
T.B., N.S., and A.S. were partially supported by BSF grant no. 2019752, and NSF grant no. 2009753.  
W.L. and N.S. were partially supported by BSF grant no. 2018230.
A.S. was also partially supported by NIH/NIGMS award 1R01GM136780-01, award FA9550-17-1-0291 from AFOSR, the Simons Foundation Math+X Investigator Award, and the Moore Foundation Data-Driven Discovery Investigator Award.
T.B. was also partially supported by the Zimin Institute for Engineering Solutions Advancing Better Lives.

\bibliographystyle{plain}

\appendix


\section{Proof of Proposition~\ref{prop:symmetry_in_circulant_sigma} } \label{app:R_invariant}

 
{The} first lemma refines the condition of Lemma~\ref{lemma:unique_perm} {in terms of invariant subspaces. Recall that a subspace $V$ is $R$-invariant if $R(V) \subset V$.}
 \begin{lemma} \label{lemma:commutes}
 	Let $S$ be a symmetric matrix. Then, $R$ and $S$ commute {if and only if} every eigenspace is $R$-invariant. 
 \end{lemma}
 \begin{proof}
 	Since $S$ is symmetric, we can decompose the space into real eigenspaces. 
 	First direction: if $R$ and $S$ commute then for any eigenvector $v$ of $S$ with eigenvalue $\lambda$, we have 
 	\[ S (Rv) = R S v = \lambda (Rv). \]
 	Namely, $Rv$ is also in the eigenspace. {Second direction:} if any eigenspace $V_\lambda$ is $R$-invariant then for any $v \in V_\lambda$  we can write $Rv$ in terms of the basis of $V_\lambda$, and to have
 	\[ Rv = \sum_i \alpha_i v_i \implies S (Rv) = S\left(\sum_i \alpha_i v_i \right) =  \sum_i \alpha_i \lambda  v_i  = \lambda Rv = R (Sv)  . \qedhere \] 
 \end{proof}

\begin{lemma}
\label{lemma:circulant}
Let $\Sigma^{-1}$ be a circulant matrix. Then 
\begin{enumerate}
\item
$\Sigma^{-1}$ commutes with any cyclic permutation matrix $R$. 

\item
Conversely, if each eigenvalue of $\Sigma^{-1}$ has multiplicity $1$  and $\Sigma^{-1}$ commutes with a permutation matrix $R$, then $R$ must be a cyclic permutation.
\end{enumerate}
\end{lemma}

\begin{proof}
Let $W$ denote the DFT matrix, with entries  $W_{j,\ell} = \frac{1}{\sqrt{n}} \omega^{(\ell-1)(j-1)}$,  $\omega=e^{-2 \pi \I /n}$, $1 \le j,\ell \le n$. Then the columns $w_0,\dots,w_{n-1}$ of $W$ are the eigenvectors of $\Sigma^{-1}$. Furthermore, if $R$ is the cyclic permutation matrix corresponding to the permutation $j \mapsto j+k \mod n$, then $R w_{\ell} = \omega^{(\ell-1)k} w_\ell$. Consequently, $R$ preserves each eigenspace of $\Sigma^{-1}$, and so by Lemma~\ref{lemma:commutes} $R$ and $\Sigma^{-1}$ commute. This proves the first statement.

For the converse, suppose each eigenvalue of $\Sigma^{-1}$ has multiplicity $1$, and take any permutation matrix $R$ that commutes with $\Sigma^{-1}$. From Lemma~\ref{lemma:commutes},  $R$ must leave each eigenspace of $\Sigma^{-1}$ fixed; consequently, for each eigenvector $w_\ell$ of $\Sigma^{-1}$ there is a scalar $\alpha_\ell$ so that $R w_\ell = \alpha_\ell w_\ell$. Suppose the permutation corresponding to $R$ sends index $1$ to index $k+1$; then
\begin{align}
\frac{1}{\sqrt{n}} = W_{1,\ell} = (R w_\ell)_{k+1} = \alpha_\ell W_{k+1,\ell} = \alpha_\ell \frac{1}{\sqrt{n}}\omega^{(\ell-1)k},
\end{align}
and so $\alpha_\ell = \omega^{-(\ell-1)k}$. Since $R w_\ell = \alpha_\ell w_\ell$, for any index $j$,
\begin{align}
(R w_\ell)_j = \omega^{-(\ell-1)k} W_{j,\ell} = \omega^{-(\ell-1)k} \frac{1}{\sqrt{n}}\omega^{(\ell-1)(j-1)}
= \frac{1}{\sqrt{n}} \omega^{(\ell-1)(j-k - 1)}
= W_{j-k,\ell},
\end{align}
meaning that $R$ cyclically shifts the entries of $w_\ell$ by $k$. Since this holds for all basis vectors $w_\ell$, $R$ is a cyclic shift.
\end{proof}

We may now prove Proposition~\ref{prop:symmetry_in_circulant_sigma}. The first statement of Proposition~\ref{prop:symmetry_in_circulant_sigma} is identical to the first statement of Lemma \ref{lemma:circulant}. For the second statement, Lemma~\ref{lemma:unique_perm} tells us that for almost every signal $x$, the quadratic form $y^T  \Sigma^{-1} y$ takes on distinct values on each equivalence class of vectors in the orbit $G_{\Pi,L}x$ (where two vectors are equivalent if one is a cyclic shift of the other). Indeed, for any two permutation matrices $R_1$ and $R_2$, the set $\left\lbrace x \mid x^T R_1^T \Sigma^{-1} R_1 x = x^T R_2^T \Sigma^{-1} R_2 x \right\rbrace$ {has measure zero} if and only if $R_2 R_1^T$ does not commute with $\Sigma^{-1}$; from Lemma~\ref{lemma:circulant}, this latter condition is equivalent to  $R_2 R_1^T$ not being a cyclic permutation. Since there are only finitely many permutation matrices, the set
\begin{align}
 \left\lbrace x \mid  \exists R_1, R_2 \text{ s.t. } R_2 R_1^T \text{ not cyclic and } x^T R_1^T \Sigma^{-1} R_1 x = x^T R_2^T \Sigma^{-1} R_2 x   \right\rbrace
\end{align}
also has measure $0$. Consequently, for almost every signal $x$, the equality $x^T R_1^T \Sigma^{-1} R_1 x = x^T R_2^T \Sigma^{-1} R_2 x$ can hold only when $R_2 R_1^T$ is cyclic, i.e.\  when $R_1 x$ and $R_2 x$ are in the same equivalence class of $G_{\Pi,L}x$.

Because the quadratic form $y^T \Sigma^{-1} y$ takes on distinct values on each equivalence class in $G_{\Pi,L}x$ for almost every $x$, it immediately follows that for almost every $x$ the minimum of $y^T \Sigma^{-1} y$ over equivalence classes in $G_{\Pi,L}x$ is unique. This is the desired result.

\end{document}